\documentclass{cta-author}%%%%where cta-author is the template name
\pdfoutput=1
\usepackage{amsmath,amssymb,amsthm,amsfonts,xpatch}
\usepackage{subcaption}
\usepackage{relsize,steinmetz}
\usepackage{pstricks,pifont}
\usepackage[english]{babel}
\usepackage{balance}

\usepackage[inline]{enumitem}
\setlist[enumerate,1]{}

\usepackage{tikz}
\usetikzlibrary{shapes,arrows,positioning,calc}

\usepackage[nameinlink]{cleveref}
\theoremstyle{plain}

\newtheorem{theorem}{Theorem}
\newtheorem*{proof}{Proof}
\newtheorem{lemma}{Lemma}

{}
{}

\theoremstyle{definition}
\newtheorem{definition}{Definition}

\newcommand*{\QEDB}{\hfill\ensuremath{\square}}%

\crefname{equation}{Eqn.}{Eqns.}
\crefname{ineq}{Ineq.}{Ineqs.}
\creflabelformat{ineq}{#2{\upshape(#1)}#3} 
\crefname{figure}{Fig.}{Figs.}
\crefname{table}{Table}{Tables}
\crefname{theorem}{Theorem}{Theorems}
\crefname{lemma}{Lemma}{Lemmas}
\let\oldLemma\lemma
\renewcommand{\lemma}{%
  \crefalias{Theorem}{Lemma}% Theorem counter now looks like Lemma
  \oldLemma}
\crefname{section}{Section}{Sections}

\makeatletter
\pgfdeclareshape{satnode}{
\inheritsavedanchors[from={rectangle}]
\inheritbackgroundpath[from={rectangle}]
\inheritanchorborder[from={rectangle}]
\foreach \x in {center,north east,north west,north,south,south east,south west}{
\inheritanchor[from={rectangle}]{\x}
}
\foregroundpath{
\pgfpointdiff{\northeast}{\southwest}
\pgf@xa=\pgf@x \pgf@ya=\pgf@y
\northeast
\pgfpathmoveto{\pgfpoint{0}{0.45\pgf@ya}}
\pgfpathlineto{\pgfpoint{0}{-0.45\pgf@ya}}
\pgfpathmoveto{\pgfpoint{0.45\pgf@xa}{0}}
\pgfpathlineto{\pgfpoint{-0.45\pgf@xa}{0}}
\pgfpathmoveto{\pgfpointadd{\southwest}{\pgfpoint{-0.35\pgf@xa}{-0.3\pgf@ya}}}
\pgfpathlineto{\pgfpointadd{\southwest}{\pgfpoint{-0.5\pgf@xa}{-0.3\pgf@ya}}}
\pgfpathlineto{\pgfpointadd{\northeast}{\pgfpoint{-0.5\pgf@xa}{-0.3\pgf@ya}}}
\pgfpathlineto{\pgfpointadd{\northeast}{\pgfpoint{-0.25\pgf@xa}{-0.3\pgf@ya}}}
{
   \pgftransformshift{\pgfpointadd{\northeast}{\pgfpoint{-0.25\pgf@xa}{-0.3\pgf@ya}}}
   \pgftransformscale{0.5}
   \pgfsetcolor{black}
   \pgftext[left]{$\Delta$}
}
{
   \pgftransformshift{\pgfpointadd{\southwest}{\pgfpoint{0\pgf@xa}{-0.3\pgf@ya}}}
   \pgftransformscale{0.5}
   \pgfsetcolor{black}
   \pgftext[left]{$-\Delta$}
}
\pgfusepath{stroke}
}
}
\makeatother

\allowdisplaybreaks

\begin{document}
%\supertitle{Brief Paper: This paper is a preprint of a paper submitted to \emph{IET Control Theory \& Applications}. If accepted, the copy of record will be available at the IET Digital Library}
\supertitle{Brief Paper: This paper is a postprint of a paper submitted to and accepted for publication in IET Control Theory \& Applications and is subject to Institution of Engineering and Technology Copyright. The copy of record is available at the IET Digital Library}
\title{Consensus of second order multi-agents with actuator saturation and asynchronous time-delays}

\author{\au{Venkata Karteek Yanumula$^{1}$} \au{Indrani Kar$^1$} \au{Somanath Majhi$^1$}}

\address{\add{1}{Dept. of EEE, IIT Guwahati-781039, India.}
\email{yanumula@iitg.ernet.in}}

\begin{abstract}
\looseness=-1 This article presents the consensus of a saturated second order multi-agent system with non-switching dynamics that can be represented by a directed graph. The system is affected by data processing (input delay) and communication time-delays that are assumed to be asynchronous. The agents have saturation nonlinearities, each of them is approximated into separate linear and nonlinear elements. Nonlinear elements are represented by describing functions. Describing functions and stability of linear elements are used to estimate the existence of limit cycles in the system with multiple control laws. Stability analysis of the linear element is performed using Lyapunov-Krasovskii functions and frequency domain analysis. A comparison of pros and cons of both the analyses with respect to time-delay ranges, applicability and computation complexity is presented. Simulation and corresponding hardware implementation results are demonstrated to support theoretical results.
\end{abstract}

\maketitle
\tikzstyle{block} = [draw, fill=blue!5, rectangle, minimum height=3em, minimum width=3em]
\tikzstyle{sum} = [draw, fill=blue!5, circle, node distance=1cm]

\section{Introduction}\label{intro}
In the recent past, multi-agent systems have attracted a lot of attention due to their wide range of application in robotics, unmanned air and underwater vehicles, automated traffic signal control, wireless sensor networks, etc.. One of the most important problems in coordinated control is consensus of a multi-agent system, which deals with algorithms required for the convergence of agents \cite{Jadbabaie2003, Saber2004}. After the initial study by Vicsek {\it et al.}\cite{Vicsek1995} on  self-ordered motions in biologically motivated particles, Jadbabaie {\it et al.} \cite{Jadbabaie2003} gave the theoretical explanation. Olfati-Saber {\it et al.} \cite{Saber2004} provided mathematical analysis of consensus behaviour in linear first order agents with time-delay using graph theory concepts. Multi-agent consensus problems with higher order agents,  switching topologies, time-delays, nonlinearities etc., started receiving more attention \cite{Cao2013,Wang2016} after the initial results given by authors in  \cite{Jadbabaie2003,Saber2004}.

However, the majority of control laws are designed to solve consensus problems in linear multi-agent systems \cite{Jadbabaie2003, Saber2004, Xiao2008, RenW2008, Hu2010, Munz2010, Meng2011, Zhang2013, Meng2016}. For linear systems, it has been shown that eigenvalues of graph laplacian play an important role in estimating whether the network of agents converge. Since nonlinearities are unavoidable in most of the practical applications, nonlinear agents and the corresponding control laws are being considered recently \cite{Yu2010, Liu2013, Li2015}. Mobile agents generally have limited capability due to  factors like actuator saturation, moment of inertia, maximum limit on velocity, etc.. Actuator saturation is frequently encountered due to limitations in hardware. Some of the researchers focused on consensus in multi-agent systems with saturation in first order \cite{Li2011} and second order agents \cite{Meng2013, Wei2014, Chu2015, Su2015, Cui2016, You2016}. 

Apart from eigenvalues of graph laplacian, time-delays play major role in stability of multi-agent network. In practical applications, time-delays are inevitable and are classified into two categories; communication and input time-delays. The amount of time taken by agents to communicate is defined as communication time-delay and the amount of time taken by agents to process the information received from other agents is called input time-delay. Olfati-Saber {\it et al.}\cite{Saber2004} started the analysis of time-delay effects on multi-agent systems and gave an upper bound for first order agents considering constant uniform communication and input time-delays. Later on it was extended to systems with first order agents and uniform time-varying delay\cite{Xiao2008} multiple delays \cite{Sun2008, Munz2010}, second order agents with constant time-delays \cite{Hu2010, Lin2010} and system with second order agents with non-uniform delay \cite{Zhang2013}. 

The majority of research is confined to linear agents with time-delays and recently nonlinear agents with time-delays are receiving attention \cite{Liu2013,You2016}. Furthermore, nonlinearities are common in mobile agents and actuator saturation is the most frequent hard nonlinearity affecting them. For example, the acceleration of an agent is constant over certain range and cannot be maintained after the agent attains its maximum velocity. Recently, saturation nonlinearity is receiving considerable attention. Li {\it et al.} \cite{Li2011} considered a first order system with saturation and without time-delays. For second order agents with saturation and without time-delay, a differential gain feedback control is used by authors in \cite{Meng2013,Wei2014}. Adaptive control laws with an observer are used by Chu {\it et al.} \cite{Chu2015} and nonlinear agents are considered by Cui {\it et al.} \cite{Cui2016}. The effects of synchronous time-delays are taken into consideration by You {\it et al.} \cite{You2016} for a network of second order saturated agents.  

It is evident from the literature that, there is very little focus on consensus of second order saturated multi-agent system with asynchronous communication and input time-delays. In this contribution, a multi-agent system is considered with asynchronous time-delays and hard saturation nonlinearities. The objective of the article is to extend the results of Liu {\it et al.} \cite{Liu2013} for velocity saturated nonlinear multi-agent system with time-delays using a different approach. Describing function analysis \cite{slotine_chap} is used to break agents into approximate linear and nonlinear elements, with the nonlinear element represented by an appropriate describing function. The difference among position states and velocity states of agents is defined as error dynamics. The system achieves consensus when the error dynamics are asymptotically stable. Here, the existence of limit cycles in the multi-agent system is estimated with the help of describing functions and stability of linear element. Lyapunov-Krasovskii functions and frequency domain analysis are used to prove the stability of linear element and further estimate the stability of limit cycles. Consensus is achieved when there are no limit cycles. Some necessary and sufficient conditions for consensus in terms of linear matrix inequalities and explicit expressions are derived.
The major contributions of the paper can be summarised as, 
\begin{enumerate*}
\item Deriving various conditions for four consensus control laws with asynchronous time-delays;
\item Describing function analysis is used to estimate the limit cycle behaviour of the system;
\item Stability analysis of the linear element using Lyapunov-Krasovskii and frequency domain approaches is performed;
\item A comparison of pros and cons of both the stability analyses is presented;
\item Simulations and further validation of results on a four-agent and a five-agent networks are demonstrated to support theoretical analysis.
\end{enumerate*}

The rest of the paper is organised as follows, \cref{sec2} explains graph theory preliminaries. \cref{sec3} elaborates the system model with four control laws given in \cref{eqn1c,eqn1d,eqn1e,eqn1f}. Stability analysis is performed using Lyapunov-Krasovskii functions for control laws in \cref{eqn1c,eqn1d,eqn1e,eqn1f}, using the Nyquist stability criterion for control laws in \cref{eqn1c,eqn1d}. Furthermore, simulation and implementation of the control laws on two networks are explained. Depiction of results and comparison of the two stability procedures are performed in \cref{sec4}.
\section{Preliminaries}\label{sec2}
\subsection{Graph theory}\label{2a}
Graph theory is widely used to study multi-agent systems. A network of agents and the underlying communication topology can be represented by a graph $\mathcal{G}=(\mathcal{V},\mathsf{E}, \mathcal{A})$. If the communication among agents could be unidirectional, a directed graph is used to describe the multi-agent network. The vertex set $ \mathcal{V} = \{\mathsf{v}_1,\mathsf{v}_2,....,\mathsf{v}_n\}$ where vertices are analogous to agents and an edge set $\mathsf{E} = \{ (i,j):i,j\in \mathcal{V}\}$ where edges are analogous to the branches of directed network with $(i,j)$ representing information flowing from $j^{th}$ vertex to $i^{th}$. Edge set has distinct ordered pairs of vertices which depict existence and direction of information flow among the vertices. An adjacency matrix $\mathcal{A} = (a_{ij} )_{n \times n}$ also represents communication topology with $a_{ij} = 1$ if $(i,j) \in \mathsf{E}$ and $a_{ij} = 0$ otherwise. A weighted adjacency matrix will have entries other than zero and unity weights depending on the assumptions of cost of communication. If there exists at least one vertex which has a directed path to all the other vertices, the graph is said to form a spanning tree and if all the vertices have directed paths to all the other agents, it is called strongly connected. A spanning tree condition is a necessary condition for consensus but not sufficient when time-delays and higher order systems are involved \cite{Cao2013,Wang2016}. The sum of weights of inward branches at a vertex is called in-degree $d_{in}(v_i)$ and the weight sum of outward branches is called out-degree of the vertex $d_{out}(v_i)$.
\subsection{Notations}\label{sec2b}
The following notations are used throughout the paper, $\mathbb{R}^{n}$ represents an $n$-dimensional Euclidean space. $\mathbb{R}^{m \times n}$ represent a space of $m \times n$ matrices. Position and velocity of $n$ agents are represented by $\boldsymbol{x}=[x_1\ x_2\ ... \ x_n]^{T}$ and $\boldsymbol{\dot{x}}=\boldsymbol{v}$ respectively. $\boldsymbol{X}=[X_1\ X_2\ ...\ X_n\ X_{n+1}\ ...\ X_{2n}]^{T}$ represent the states of a multi-agent system with $[X_1\ X_2\ ...\ X_n]^T = \boldsymbol{x}$ and $[X_{n+1}\ ...\ X_{2n}]^T =\boldsymbol{\dot{x}}$. $I_n$ and $I_{2n}$ represent identity matrices of sizes $n \times n$ and $2n \times 2n$ respectively. $\bold{1}_n$ is a vector ones of size $1 \times n$. For $\{A,B\} \in \mathbb{R}^{n \times n}$, if $A \succcurlyeq B$, then $A-B$ is positive semidefinite; if $A \succ B$, then $A-B$ is positive definite. $D$ represents a matrix with diagonal elements as row-sum of adjacency matrix $\mathcal{A}$ and rest of the elements as zero. A matrix $\tilde{\mathcal{A}}$ is defined with elements $\tilde{a}_{ij} = \dfrac{a_{ij}}{\sum_{j=1}^n a_{ij}}$ and $[\tilde{\lambda}_1, \tilde{\lambda}_2, ... , \tilde{\lambda}_n]$ are the eigenvalues of matrix $\tilde{\mathcal{A}}$.

\section{System model and analysis}\label{sec3}
Consider a multi-agent network of homogeneous second order agents with $i^{th}$ agent dynamics given in \cref{eqn1},
\begin{equation}\label{eqn1}
\begin{array}{lll}
\dot{x}_i\left(t\right) &=& sat\left(\hat{v}_i\left(t\right)\right) \\
\dot{\hat{v}}_i\left(t\right) &=& u_i \left(t\right)
\end{array}
\end{equation}
For mobile agents, the position of an agent is represented by $x_i\left(t\right)$ and the velocity by $\dot{x}_i\left(t\right)$. Various control protocols used in the analysis are given in \cref{eqn1c,eqn1d,eqn1e,eqn1f}. It is assumed that $\forall t \in (-\infty,0]$, $x_i\left(t\right) = x\left(0\right)$ and $\hat{v}_i\left(t\right)=0$.  Saturation nonlinearity used in the system is defined in \cref{satu} with $\pm \Delta$ as bounds.

\begin{equation}\label{satu}
  sat\left(\alpha\right)=
  \begin{cases}
    -\Delta, & \text{if}\ \alpha \leq -\Delta \\
    \alpha, & \text{if}\ -\Delta < \alpha < \Delta \\
    \Delta, & \text{if}\ \alpha \geq \Delta
  \end{cases}
\end{equation}
\begin{figure}[!h]
\centering
\begin{tikzpicture}[auto, node distance=2cm,>=latex']
    \node [block] (int1) {$\mathlarger\int$};
    \node[satnode,minimum height=1cm,minimum width=1.5cm,right of=int1,fill=blue!5,draw,node distance=2cm] (sat) {};
    \node [block, right of=sat, node distance=2cm] (int2) {$\mathlarger\int$};
    \node [coordinate, right of=int2] (output) {};
    \node [block, below left of=sat,align=center] (comm) {Communication \\ \& process };
    \node [coordinate,left of=int1] (fb) {};
    \draw [thick,->] (int1) -- node[name=v1] {$\hat{v}_i$} (sat);
    \draw [thick,->] (sat) -- node[name=v2] {$v_i$} (int2);
    \draw [thick,->] (comm.320) -- node [near end, right, name=x1] {$(a_{ji})(x_i,v_i)$} ++(0cm,-1cm);
    \draw [thick,<-] (comm.220) -- node [near end, left, name=x2] {$(a_{ij})(x_j,v_j)$} ++(0cm,-1cm);    
    \draw [thick,->] (int2) -- node [name=x] {$x_i$}(output);
    \draw [thick,->] (x) |- (comm.350);
    \draw [thick,->] (v2) |- (comm.10);
    \draw [thick,->] (comm) -| (fb) -- node [near start] {$u_i$} (int1);
\end{tikzpicture}
\caption{Block diagram of $i^{th}$ agent.}
\label{bd1}
\end{figure}
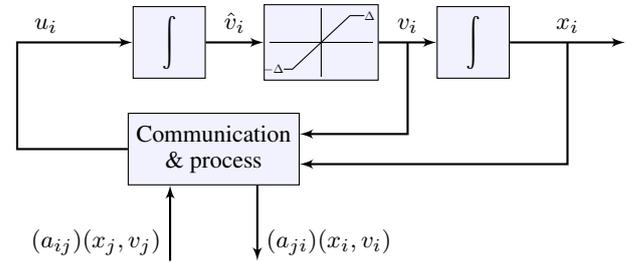
The $i^{th}$ agent dynamics are depicted using a block diagram given in \cref{bd1}. Using the concepts of describing function to estimate limit cycles \cite{slotine_chap}, the system can be approximately transformed as shown in  \cref{bd2}. Since a single-valued nonlinearity is considered, its approximate describing function for the saturation is given in \cref{df} \cite{slotine_chap}, 
\begin{equation}\label{df}
N(A) = \frac{2}{\pi}\left[\arcsin\left(\frac{\Delta}{A}\right) + \frac{\Delta}{A}\sqrt{1-\frac{\Delta^2}{A^2}}\right]
\end{equation}
where, the limit cycles' amplitude is represented by $A$. 

The describing function is real valued and $-1/N(A) \in [-1,\infty)$, it can be estimated that the limit cycles are stable when the transfer function of linear element in \cref{bd2} encircles $(-1,0)$ in a complex plane. In other words, limit cycles are exhibited when the linear element is unstable in the multi-agent system. Stability analysis of the linear element is performed using Lyapunov-Krasovskii approach in \cref{lyap} and Nyquist stability approach given in \cref{freq}.

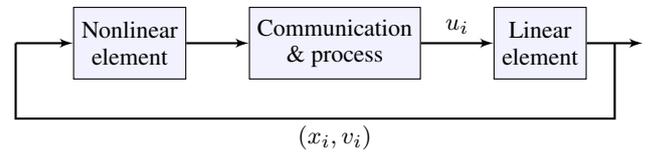
\begin{figure}[!h]
\centering
%\resizebox{0.45\textwidth}{!}{
\begin{tikzpicture}[auto, node distance=2cm,>=latex']
    \node [coordinate, name=input] {};
    %\node [sum, right of=input] (sum) {$+$};    
    \node [block, right of=input,node distance=1.5cm, align=center] (nonlin) {Nonlinear \\ element};
    \node [block, right of=nonlin,node distance=2.7cm, align=center] (comm) {Communication \\ \& process};
    \node [block, right of=comm,node distance=2.7cm, align=center] (lin) {Linear \\ element};
    \node [coordinate, right of=lin,node distance=1.35cm] (output) {};
    \node [coordinate,below of=comm,node distance=1cm] (fb) {};
    \node [coordinate,above of=nonlin,node distance=1cm] (ab) {};
    %\node [above of=ab,node distance=0.25cm] (ab2) {$(a_{ij})(x_j, v_j)$};
    \node [below of=fb,node distance=0.25cm] (fb2) {$(x_i, v_i)$};
    %\node (symb1) at ($(sum)+ (-70:0.65)$){$-$}; 
    %\node (symb2) at ($(sum)+ (70:0.65)$){$+$};    
    \draw [thick,->] (nonlin) -- node {} (comm);
    \draw [thick,->] (comm) -- node {$u_i$} (lin);
    %\draw [thick,->] (comm.90) |- (ab) -|(sum);
    \draw [thick,->] (input) -- (nonlin);
    \draw [thick,->] (lin) -- node [name=y] {}(output);
    \draw [thick,-] (y) |- (fb) -| node [near end] {} (input);
\end{tikzpicture}%}
\caption{Rearranged block diagram of $i^{th}$ agent.}
\label{bd2}
\end{figure}
The approximate linear element given in \cref{bd2} is represented by \cref{eqn1b},
\begin{equation}\label{eqn1b}
\begin{array}{lll}
\dot{x}_i\left(t\right) &=& v_i\left(t\right) \\
\dot{v}_i\left(t\right) &=& u_i \left(t\right)
\end{array}
\end{equation}
Various control laws considered from the literature for analysis are given in \cref{eqn1c,eqn1d,eqn1e,eqn1f}, 
\begin{equation}\label{eqn1c}
\begin{split}
u_{i1}\left(t\right)& = -v_i\left(t-\tau_1\right) \\& +\frac{1}{\sum_{j=1}^n a_{ij}} \sum_{j=1}^n \big[ a_{ij}\left(x_j\left(t-\tau_2\right) - x_i\left(t-\tau_1\right)\right)\big]
\end{split}
\end{equation}
\begin{equation}\label{eqn1d}
\begin{split}
u_{i2} = \frac{1}{\sum_{j=1}^n a_{ij}}  \sum_{j=1}^n &\big[ a_{ij}\left(v_j\left(t-\tau_2\right) - v_i\left(t-\tau_1\right)\right) \\&+ a_{ij}\left(x_j\left(t-\tau_2\right) - x_i\left(t-\tau_1\right)\right) \big]
\end{split}
\end{equation}
\begin{equation}\label{eqn1e}
\begin{split}
u_{i3}\left(t\right) = & -v_i \left(t- \tau_1\right) \\& + \sum_{j=1}^n \big[ a_{ij}\left(x_j\left(t-\tau_2\right) - x_i\left(t-\tau_1\right)\right)\big]
\end{split}
\end{equation}
\begin{equation}\label{eqn1f}
\begin{split}
u_{i4} = \sum_{j=1}^n \big[ a_{ij} & \left(v_j\left(t-\tau_2\right) - v_i\left(t-\tau_1\right)\right) \\&+ a_{ij}\left(x_j\left(t-\tau_2\right) - x_i\left(t-\tau_1\right)\right) \big]
\end{split}
\end{equation}
where $\tau_1$ and $\tau_2$ represent input and communication time-delays respectively. With any of the control laws in \cref{eqn1c,eqn1d,eqn1e,eqn1f}, consensus is said to be reached if $(x_i(t)-x_j(t)) \rightarrow 0$ and $(\dot{x}_i(t)-\dot{x}_j(t)) \rightarrow 0$ $\forall \{i,j\} \in [1,\;n]$. Control laws in \cref{eqn1c,eqn1d} generate lesser magnitude of control input $u_i$ which result in slightly larger convergence time compared to the ones in \cref{eqn1e,eqn1f}. The averaging in control laws given by \cref{eqn1c,eqn1d} have better time-delay tolerance due to smaller Fiedler eigenvalue compared to control laws in \cref{eqn1e,eqn1f} at the expense of convergence time. With control laws in \cref{eqn1c,eqn1e}, the state $\dot{x}_i(t) \rightarrow 0$ when the consensus is achieved since they do not consider difference in velocity. State $\dot{x}_i(t) \rightarrow 0$ is not guaranteed with control laws in \cref{eqn1d,eqn1f}. 
\subsection{Lyapunov-Krasovskii approach}\label{lyap}

Consider the linear element represented in \cref{eqn1b,eqn1c,eqn1d,eqn1e,eqn1f}, which can be represented as given in \cref{eqn_l1}. 
\begin{equation}\label{eqn_l1}
\boldsymbol{\dot{X}}\left(t\right) = \mathcal{A}_0\boldsymbol{X}\left(t\right) + \mathcal{A}_1\boldsymbol{X}\left(t-\tau_{1}\right) +\mathcal{A}_2\boldsymbol{X}\left(t-\tau_{2}\right)
\end{equation}
Where $\mathcal{A}_0$, $\mathcal{A}_1$ and $\mathcal{A}_2$ are as given in \cref{eqn_l2,eqn_l3,eqn_l4,eqn_l5}. \\
For $u_{i1}$ given in \cref{eqn1c},
\begin{equation}\label{eqn_l2}
\begin{split}
& \mathcal{A}_0 = \left[\begin{matrix}
0_{n \times n} & I_n \\
0_{n \times n} & 0_{n \times n}
\end{matrix} \right];\, \mathcal{A}_1=\left[\begin{matrix}
0_{n \times n} & 0_{n \times n} \\
-I_n & -I_n
\end{matrix} \right];\,\\& \mathcal{A}_2=\left[\begin{matrix}
0_{n \times n} & 0_{n \times n} \\
\widetilde{\mathcal{A}} & 0_{n \times n}
\end{matrix} \right]
\end{split}
\end{equation}
For $u_{i2}$ given in \cref{eqn1d}, 
\begin{equation}\label{eqn_l3}
\begin{split}
& \mathcal{A}_0 = \left[\begin{matrix}
0_{n \times n} & I_n \\
0_{n \times n} & 0_{n \times n}
\end{matrix} \right];\, \mathcal{A}_1=\left[\begin{matrix}
0_{n \times n} & 0_{n \times n} \\
-I_n & -I_n
\end{matrix} \right];\,\\& \mathcal{A}_2=\left[\begin{matrix}
0_{n \times n} & 0_{n \times n} \\
\widetilde{\mathcal{A}} & \widetilde{\mathcal{A}}
\end{matrix} \right]
\end{split}
\end{equation}
For $u_{i3}$ given in \cref{eqn1e}, 
\begin{equation}\label{eqn_l4}
\begin{split}
& \mathcal{A}_0 = \left[\begin{matrix}
0_{n \times n} & I_n \\
0_{n \times n} & 0_{n \times n}
\end{matrix} \right];\, \mathcal{A}_1=\left[\begin{matrix}
0_{n \times n} & 0_{n \times n} \\
-D & -I_n
\end{matrix} \right];\,\\& \mathcal{A}_2=\left[\begin{matrix}
0_{n \times n} & 0_{n \times n} \\
\mathcal{A} & 0_{n \times n}
\end{matrix} \right]
\end{split}
\end{equation}
For $u_{i4}$ given in \cref{eqn1f}, 
\begin{equation}\label{eqn_l5}
\begin{split}
& \mathcal{A}_0 = \left[\begin{matrix}
0_{n \times n} & I_n \\
0_{n \times n} & 0_{n \times n}
\end{matrix} \right];\, \mathcal{A}_1=\left[\begin{matrix}
0_{n \times n} & 0_{n \times n} \\
-D & -D
\end{matrix} \right];\,\\& \mathcal{A}_2=\left[\begin{matrix}
0_{n \times n} & 0_{n \times n} \\
\mathcal{A} & \mathcal{A}
\end{matrix} \right]
\end{split}
\end{equation}
Some definitions and lemmas analogous to the ones in \cite{Lin2008} are given below,
\begin{definition}
\emph{Balanced graph}: A graph is said to be balanced if in-degree equals to out-degree for all vertices in the graph, $d_{in}(v_i)=d_{out}(v_i),\, \forall i \in [1,n]$.
\end{definition}
\begin{definition}
\emph{$k$-regular graph}: It is a balanced graph with all the vertices having in-degree and out-degree equal to $k$, $d_{in}(v_i)=d_{out}(v_i)=k,\, \forall i \in [1,n]$
\end{definition}
Control laws given in \cref{eqn1c,eqn1d} make the multi-agent system behave like a system connected by 1-regular graph.
\begin{lemma}\label{lem1}
Consider $\Phi_{01} = \frac{1}{n}\left[\begin{matrix}
1_{n \times n} & 0_{n \times n} \\
0_{n \times n} & 1_{n \times n}
\end{matrix} \right]$ and $\mathcal{E} = I_{2n}-\Phi_{01}$, then the following statements hold true:
\begin{enumerate}
	\item A multi-agent system with $k$-regular graph communication topology and with inputs in \cref{eqn1c,eqn1d,eqn1e,eqn1f} produces balanced matrices $\mathcal{E}\left(\mathcal{A}_0 + \mathcal{A}_1 + \mathcal{A}_2\right)$, $\mathcal{E}\mathcal{A}_0$, $\mathcal{E}\mathcal{A}_1$ , $\mathcal{E}\mathcal{A}_2$ and $\mathcal{E}\left(\mathcal{A}_1 + \mathcal{A}_2\right)$ with maximum rank $2n-2$ and eigenvalues $0$ of multiplicity atleast two.
	\item A multi-agent system with a spanning tree in communication topology and with inputs in \cref{eqn1c,eqn1d} produces balanced matrices $\mathcal{E}\left(\mathcal{A}_0 + \mathcal{A}_1 + \mathcal{A}_2\right)$, $\mathcal{E}\mathcal{A}_0$, $\mathcal{E}\mathcal{A}_1$ , $\mathcal{E}\mathcal{A}_2$ and $\mathcal{E}\left(\mathcal{A}_1 + \mathcal{A}_2\right)$ with maximum rank $2n-2$ and eigenvalues $0$ of multiplicity atleast two.
\end{enumerate}
\end{lemma}
\begin{definition}
\emph{Balanced matrix}: A square matrix $\mathcal{M} \in R^{n \times n}$ is said to be balanced iff $\mathcal{M} \bold{1}_n^T = 0$ and $\bold{1}_n \mathcal{M} = 0$.
\end{definition}
\begin{lemma} \label{lem2}
Consider $\Phi_{01} = \frac{1}{n}\left[\begin{matrix}
1_{n \times n} & 0_{n \times n} \\
0_{n \times n} & 1_{n \times n}
\end{matrix} \right]$ and $\mathcal{E} = I_{2n}-\Phi_{01}$, then the following statements hold true for $k$-regular graph with inputs in \cref{eqn1c,eqn1d,eqn1e,eqn1f} and for spanning tree graph with inputs in \cref{eqn1c,eqn1d}:
\begin{enumerate}
\item $\mathcal{E}\left(\mathcal{A}_0 + \mathcal{A}_1 + \mathcal{A}_2\right)$ is a balanced matrix with rank $2n-2$ and eigenvalues $0$ of multiplicity $2$.
\item Matrices $\mathcal{E}\mathcal{A}_0$, $\mathcal{E}\mathcal{A}_1$ , $\mathcal{E}\mathcal{A}_2$ and $\mathcal{E}\left(\mathcal{A}_1 + \mathcal{A}_2\right)$ are all balanced with eigenvalues $0$ of multiplicity atleast $2$.
\item There is a matrix $U$, an orthogonal matrix of eigenvectors of $\mathcal{E}$ which satisfies, \\
$U^T\mathcal{E}U= \left[\begin{matrix}
\widetilde{\mathcal{E}}_{\left(2n-2\right) \times 2} & 0_{\left(2n-2\right) \times 2} \\
0_{2 \times \left(2n-2\right)} & 0_{2 \times 2}
\end{matrix} \right]$
\item Let $\mathcal{E}\mathcal{A}_0 = \mathcal{F}_0$, $\mathcal{E}\mathcal{A}_1 = \mathcal{F}_1$ and $\mathcal{E}\mathcal{A}_2 = \mathcal{F}_2$. $\mathcal{E}$, $\mathcal{F}_0$, $\mathcal{F}_1$ and $\mathcal{F}_2$ have maximum rank $2n-2$ and with zero row sums, then,\\ $U^T\mathcal{F}_iU = \left[\begin{matrix} \widetilde{\mathcal{F}}_{i_{(2n-2) \times (2n-2)}} & 0_{_{(2n-2) \times 2}} \\
0_{_{2 \times (2n-2)}} & 0_{_{2 \times 2}}
\end{matrix}\right]$, $i \in [0,2]$.
\item Also, for cases of $\tau_1=\tau_2>0$ and $\tau_1=\tau_2=0$,   \\
$U^T\mathcal{E}\left(\mathcal{A}_1 + \mathcal{A}_2\right)U = \left[\begin{matrix}
(\widetilde{\mathcal{F}}_{1}+\widetilde{\mathcal{F}}_{2})_{_{(2n-2) \times (2n-2)}} & 0_{_{\left(2n-2\right) \times 2}} \\
0_{_{2 \times \left(2n-2\right)}} & 0_{_{2 \times 2}}
\end{matrix} \right]$ \\
$U^T\mathcal{E}\left(\mathcal{A}_0 + \mathcal{A}_1 + \mathcal{A}_2\right)U = \left[\begin{matrix}
(\widetilde{\mathcal{F}}_{0}+\widetilde{\mathcal{F}}_{1}+\widetilde{\mathcal{F}}_{2})& 0_{_{\left(2n-2\right) \times 2}} \\
0_{_{2 \times \left(2n-2\right)}} & 0_{_{2 \times 2}}
\end{matrix} \right]$
\end{enumerate}
\end{lemma}

Let the difference in position and velocity among the agents be assumed as error $\boldsymbol{\Psi}$, each element of $\boldsymbol{\Psi}$ is given by \cref{eqn_l5a}
\begin{equation}\label{eqn_l5a}
\Psi_i = \begin{cases}
\frac{1}{n} \sum\limits_{j=1}^n \left(X_i-X_j\right) & \forall i \in [1,\; n] \\ 
\frac{1}{n} \sum\limits_{j=n+1}^{2n} \left(X_i-X_j\right) & \forall i \in [n+1,\; 2n]
\end{cases}
\end{equation}
From the assumption in \cref{lem1}, 
\begin{equation}\label{eqn_l5b}
\boldsymbol{\Psi}=\mathcal{E}\boldsymbol{X}
\end{equation}
\begin{lemma}\label{lem3}
When error $\boldsymbol{\Psi} \rightarrow 0$, then $x_i \rightarrow x_j$ and $v_i \rightarrow v_j$. Conversely when $x_i \rightarrow x_j$ and $v_i \rightarrow v_j$, then $\boldsymbol{\Psi} \rightarrow 0$.
\end{lemma}
\begin{proof}
\normalfont
Consider a matrices,
\begin{equation}\label{eqn_l5c}
\gamma_{_{n \times n}} = \begin{bmatrix}1& -1& 0& ... & 0\\ 0 & 1& -1& ... & 0 \\ ... & ... & ... & ... & ... \\ ... & ... & ... & ... & ...\\ 0& ... & 0& 1& -1\\ -1& 0& ...& 0& 1 \end{bmatrix}
\end{equation}
\begin{equation}\label{eqn_l5d}
\Gamma_{_{2n \times 2n}} = \begin{bmatrix} \gamma_{_{n \times n}} & 0_{_{n \times n}} \\ 0_{_{n \times n}} & \gamma_{_{n \times n}} \end{bmatrix}
\end{equation}
Multiplying with $\Gamma$ on both sides of \cref{eqn_l5b},
\begin{equation}\label{eqn_l5e}
\begin{bmatrix}
\Psi_1 - \Psi_2 \\ \Psi_2 - \Psi_3 \\ .\\. \\ \Psi_n - \Psi_1 \\ \Psi_{n+1} - \Psi_{n+2} \\ \Psi_{n+2} - \Psi_{n+3} \\ .\\. \\ \Psi_{2n} - \Psi_{n+1}
\end{bmatrix} = \begin{bmatrix}
X_1 - X_2 \\ X_2 - X_3 \\ .\\. \\ X_n - X_1 \\ X_{n+1} - X_{n+2} \\ X_{n+2} - X_{n+3} \\ .\\. \\ X_{2n} - X_{n+1}
\end{bmatrix}
\end{equation}
When $\boldsymbol{\Psi} \rightarrow 0$, left side of \cref{eqn_l5e} becomes $0_{_{2n \times 1}}$. Which implies, $X_i \rightarrow X_j, \; \forall \{i,j\} \in [1,n]$ and $ X_i \rightarrow X_j, \; \forall \{i,j\} \in [n+1,2n]$.\\
From \cref{eqn_l5a}, when $x_i \rightarrow x_j$ and $v_i \rightarrow v_j$, then $\boldsymbol{\Psi} \rightarrow 0$. \QEDB
\end{proof} 
A control input is said to have solved the consensus problem in a globally asymptotic manner when $x_i \rightarrow x_j$ and $v_i \rightarrow v_j$, in other words, $\boldsymbol{\Psi} \rightarrow 0$. Stability of linear element with the control inputs estimates the existence of limit cycles in the system.
\begin{theorem}\label{th_lyp}
Consider the linear element in \cref{eqn1b} with time-delays $(\tau_1,\tau_2) \geq 0, \tau_1 \leq \tau_2 $. The control inputs for a $k$-regular graph given in \cref{eqn1c,eqn1d,eqn1e,eqn1f} and the control inputs for a spanning tree graph given in \cref{eqn1c,eqn1d} globally asymptotically solve consensus problem, if there exist matrices $\widetilde{P}>0$, $\widetilde{Q}_1>0$, $\widetilde{Q}_2>0$, $\widetilde{Z}_i>0$, $\forall i \in [1,3]$,  $\widetilde{\mathcal{F}}_{i}$ from \cref{lem2} $\forall i \in [1,3]$ and arbitrary matrices $\{\widetilde{H}_{ij}, \widetilde{H}_{ij}^T, \widetilde{I}_{ij}, \widetilde{I}_{ij}^T, \widetilde{J}_{ij}, \widetilde{J}_{ij}^T\} \; \forall i \in [1,3]$ $\forall j \in [1,4]$ of size $(2n-2) \times (2n-2)$ such that,
\begin{equation}\label{eqn_l6}
\left[\begin{matrix} \widetilde{G}_{11} & \widetilde{G}_{12} & \widetilde{G}_{13} \\ \widetilde{G}_{12}^T & \widetilde{G}_{22} & \widetilde{G}_{23} \\ \widetilde{G}_{13}^T & \widetilde{G}_{23}^T & \widetilde{G}_{33} \end{matrix}\right] \prec 0
\end{equation}
\begin{equation}\label{eqn_l7}
\left[\begin{matrix} \widetilde{H}_{11} & \widetilde{H}_{12} & \widetilde{H}_{13} & \widetilde{H}_{14}\\ \widetilde{H}_{12}^T & \widetilde{H}_{22} & \widetilde{H}_{23} & \widetilde{H}_{24}\\ \widetilde{H}_{13}^T & \widetilde{H}_{23}^T & \widetilde{H}_{33}& \widetilde{H}_{34} \\ \widetilde{H}_{14}^T & \widetilde{H}_{24}^T & \widetilde{H}_{34}^T & \widetilde{Z}_{1} \end{matrix}\right] \succcurlyeq 0 
\end{equation}
\begin{equation}\label{eqn_l8}
\left[\begin{matrix} \widetilde{I}_{11} & \widetilde{I}_{12} & \widetilde{I}_{13} & \widetilde{I}_{14}\\ \widetilde{I}_{12}^T & \widetilde{I}_{22} & \widetilde{I}_{23} & \widetilde{I}_{24}\\ \widetilde{I}_{13}^T & \widetilde{I}_{23}^T & \widetilde{I}_{33}& \widetilde{I}_{34} \\ \widetilde{I}_{14}^T & \widetilde{I}_{24}^T & \widetilde{I}_{34}^T & \widetilde{Z}_{2} \end{matrix}\right] \succcurlyeq 0 
\end{equation}
\begin{equation}\label{eqn_l9}
\left[\begin{matrix} \widetilde{J}_{11} & \widetilde{J}_{12} & \widetilde{J}_{13} & \widetilde{J}_{14}\\ \widetilde{J}_{12}^T & \widetilde{J}_{22} & \widetilde{J}_{23} & \widetilde{J}_{24}\\ \widetilde{J}_{13}^T & \widetilde{J}_{23}^T & \widetilde{J}_{33}& \widetilde{J}_{34} \\ \widetilde{J}_{14}^T & \widetilde{J}_{24}^T & \widetilde{J}_{34}^T & \widetilde{Z}_{3} \end{matrix}\right] \succcurlyeq 0 
\end{equation}
where, 
\begin{align}
\widetilde{G}_{11}=&\widetilde{\mathcal{F}}_0^T\widetilde{P}\widetilde{\mathcal{E}} + \widetilde{\mathcal{E}}^T\widetilde{P}\widetilde{\mathcal{F}}_0 + \widetilde{\mathcal{E}}^T\widetilde{Q}_{1}\widetilde{\mathcal{E}} + \widetilde{\mathcal{E}}^T\widetilde{Q}_{2}\widetilde{\mathcal{E}} +  \widetilde{\mathcal{F}}_0^T\widetilde{\Xi}\widetilde{\mathcal{F}}_0 \nonumber \\& + \tau_1\widetilde{\mathcal{E}}^T\widetilde{H}_{11}\widetilde{\mathcal{E}} + \widetilde{\mathcal{E}}^T\widetilde{H}_{14}\widetilde{\mathcal{E}} + \widetilde{\mathcal{E}}^T\widetilde{H}_{14}^T\widetilde{\mathcal{E}} +  \tau_2\widetilde{\mathcal{E}}^T\widetilde{I}_{11}\widetilde{\mathcal{E}} + \nonumber \\& \widetilde{\mathcal{E}}^T\widetilde{I}_{14}\widetilde{\mathcal{E}} +   \widetilde{\mathcal{E}}^T\widetilde{I}_{14}^T\widetilde{\mathcal{E}} + \left(\tau_2 - \tau_1\right)\widetilde{\mathcal{E}}^T\widetilde{J}_{11}\widetilde{\mathcal{E}} \\
\widetilde{G}_{12}=&\widetilde{\mathcal{E}}^TP\widetilde{\mathcal{F}}_1 + \widetilde{\mathcal{F}}_0^T\widetilde{\Xi}\widetilde{\mathcal{F}}_1 + \tau_1\widetilde{\mathcal{E}}^T\widetilde{H}_{12}\widetilde{\mathcal{E}} -  \widetilde{\mathcal{E}}^T\widetilde{H}_{14}\widetilde{\mathcal{E}} + \nonumber \\& \widetilde{\mathcal{E}}^T\widetilde{H}_{24}^T\widetilde{\mathcal{E}} + \tau_2\widetilde{\mathcal{E}}^T\widetilde{I}_{12}\widetilde{\mathcal{E}} + \widetilde{\mathcal{E}}^T\widetilde{I}_{24}^T\widetilde{\mathcal{E}} + \nonumber \\& \left(\tau_2-\tau_1\right)\widetilde{\mathcal{E}}^T\widetilde{J}_{12}\widetilde{\mathcal{E}} + \widetilde{\mathcal{E}}^T\widetilde{J}_{14}\widetilde{\mathcal{E}} \\
\widetilde{G}_{13}=&\widetilde{\mathcal{E}}^T\widetilde{P}\widetilde{\mathcal{F}}_2 + \widetilde{\mathcal{F}}_0^T\widetilde{\Xi}\widetilde{\mathcal{F}}_2 +  \tau_1\widetilde{\mathcal{E}}^T\widetilde{H}_{13}\widetilde{\mathcal{E}} +  \widetilde{\mathcal{E}}^T\widetilde{H}_{34}^T\widetilde{\mathcal{E}} + \nonumber \\& \widetilde{\mathcal{E}}^T\widetilde{I}_{34}^T\widetilde{\mathcal{E}} + \tau_2\widetilde{\mathcal{E}}^T\widetilde{I}_{13}\widetilde{\mathcal{E}} - \widetilde{\mathcal{E}}^T\widetilde{I}_{14}\widetilde{\mathcal{E}} + \nonumber \\& \left(\tau_2-\tau_1\right)\widetilde{\mathcal{E}}^T\widetilde{J}_{13}\widetilde{\mathcal{E}} - \widetilde{\mathcal{E}}^T\widetilde{J}_{14}\widetilde{\mathcal{E}} \\
\widetilde{G}_{22}=&-\widetilde{\mathcal{E}}^T\widetilde{Q}_{1}\widetilde{\mathcal{E}} + \widetilde{\mathcal{F}}_1^T\widetilde{\Xi}\widetilde{\mathcal{F}}_1 +  \tau_1\widetilde{\mathcal{E}}^T\widetilde{H}_{22}\widetilde{\mathcal{E}} -  \widetilde{\mathcal{E}}^T\widetilde{H}_{24}\widetilde{\mathcal{E}} - \nonumber \\&  \widetilde{\mathcal{E}}^T\widetilde{H}_{24}^T\widetilde{\mathcal{E}} + \tau_2\widetilde{\mathcal{E}}^T\widetilde{I}_{22}\widetilde{\mathcal{E}} +  \left(\tau_2 - \tau_1\right)\widetilde{\mathcal{E}}^T\widetilde{J}_{22}\widetilde{\mathcal{E}} \nonumber \\& + \widetilde{\mathcal{E}}^T\widetilde{J}_{24}\widetilde{\mathcal{E}} + \widetilde{\mathcal{E}}^T\widetilde{J}_{24}^T\widetilde{\mathcal{E}} \\ 
\widetilde{G}_{23} =& \widetilde{\mathcal{F}}_1^T\widetilde{\Xi}\widetilde{\mathcal{F}}_2 + \tau_1\widetilde{\mathcal{E}}^T\widetilde{H}_{23}\widetilde{\mathcal{E}} -  \widetilde{\mathcal{E}}^T\widetilde{H}_{34}^T\widetilde{\mathcal{E}} + \tau_2\widetilde{\mathcal{E}}^T\widetilde{I}_{23}\widetilde{\mathcal{E}} - \nonumber \\& \widetilde{\mathcal{E}}^T\widetilde{I}_{24}\widetilde{\mathcal{E}} + \left(\tau_2-\tau_1\right)\widetilde{\mathcal{E}}^T\widetilde{J}_{23}\widetilde{\mathcal{E}} - \nonumber \\&  \widetilde{\mathcal{E}}^T\widetilde{J}_{24}\widetilde{\mathcal{E}} + \widetilde{\mathcal{E}}^T\widetilde{J}_{34}^T\widetilde{\mathcal{E}} \\
\widetilde{G}_{33}=&-\widetilde{\mathcal{E}}^T\widetilde{Q}_{2}\widetilde{\mathcal{E}} + \widetilde{\mathcal{F}}_2^T\widetilde{\Xi}\widetilde{\mathcal{F}}_2 + \tau_1\widetilde{\mathcal{E}}^T\widetilde{H}_{33}\widetilde{\mathcal{E}} -  \widetilde{\mathcal{E}}^T\widetilde{I}_{34}\widetilde{\mathcal{E}} \nonumber \\&  -\widetilde{\mathcal{E}}^T\widetilde{I}_{34}^T\widetilde{\mathcal{E}} + \tau_2\widetilde{\mathcal{E}}^T\widetilde{I}_{33}\widetilde{\mathcal{E}} + \left(\tau_2 - \tau_1\right)\widetilde{\mathcal{E}}^T\widetilde{J}_{33}\widetilde{\mathcal{E}} \nonumber \\&  -\widetilde{\mathcal{E}}^T\widetilde{J}_{34}\widetilde{\mathcal{E}} -  \widetilde{\mathcal{E}}^T\widetilde{J}_{34}^T\widetilde{\mathcal{E}} \\
\widetilde{\Xi} = &\tau_1\widetilde{Z}_1 + \tau_2\widetilde{Z}_2 + \left(\tau_2-\tau_1\right)\widetilde{Z}_3 
\end{align}
\end{theorem}
\begin{proof}
\normalfont
Let $P$, $Q_1$, $Q_2$, $Z_i$ $i \in [1,3]$ be balanced positive semi-definite matrices of rank $2n-2$ and $\boldsymbol{\Psi}=\mathcal{E}\boldsymbol{X}$ using $\mathcal{E}$ from \cref{lem1}. \\
The Lyapunov-Krasovskii functional is assumed as,
\begin{align}\label{eqn_l10}
V\left(\boldsymbol{\Psi}\left(t\right)\right) &= \boldsymbol{\Psi}^T\left(t\right)P\boldsymbol{\Psi}\left(t\right) + \int_{t-\tau_{1}}^{t} \boldsymbol{\Psi}^T\left(s\right)Q_{1}\boldsymbol{\Psi}\left(s\right)ds \nonumber \\& + \int_{t-\tau_{2}}^{t} \boldsymbol{\Psi}^T\left(s\right)Q_{2}\boldsymbol{\Psi}\left(s\right)ds \nonumber \\& + 
\int_{-\tau_1}^0\int_{t+\theta}^{t} \boldsymbol{\dot{\Psi}}^T\left(s\right)Z_{1}\boldsymbol{\dot{\Psi}}\left(s\right)dsd\theta \nonumber \\& + \int_{-\tau_2}^0\int_{t+\theta}^{t} \boldsymbol{\dot{\Psi}}^T\left(s\right)Z_{2}\boldsymbol{\dot{\Psi}}\left(s\right)dsd\theta \nonumber \\& + \int_{-\tau_2}^{-\tau_1}\int_{t+\theta}^{t} \boldsymbol{\dot{\Psi}}^T\left(s\right)Z_{3}\boldsymbol{\dot{\Psi}}\left(s\right)dsd\theta
\end{align}
\begin{align}\label{eqn_l11}
\dot{V}\left(\boldsymbol{\Psi}\left(t\right)\right) &= \boldsymbol{\dot{\Psi}}^T\left(t\right)P\boldsymbol{\Psi}\left(t\right) + \boldsymbol{\Psi}^T\left(t\right)P\boldsymbol{\dot{\Psi}}\left(t\right) + \nonumber \\& \boldsymbol{\Psi}^T\left(t\right)Q_1\boldsymbol{\Psi}\left(t\right) + \boldsymbol{\Psi}^T\left(t\right)Q_2\boldsymbol{\Psi}\left(t\right) - \nonumber \\& \boldsymbol{\Psi}^T\left(t-\tau_{1}\right)Q_1\boldsymbol{\Psi}\left(t-\tau_{1}\right) - \nonumber \\& \boldsymbol{\Psi}^T\left(t-\tau_{2}\right)Q_2\boldsymbol{\Psi}\left(t-\tau_{2}\right) +  \tau_1\boldsymbol{\dot{\Psi}}^T\left(t\right)Z_1\boldsymbol{\dot{\Psi}}\left(t\right) \nonumber \\& +\tau_2\boldsymbol{\dot{\Psi}}^T\left(t\right)Z_2\boldsymbol{\dot{\Psi}}\left(t\right) + \left(\tau_2-\tau_1\right)\boldsymbol{\dot{\Psi}}^T\left(t\right)Z_3\boldsymbol{\dot{\Psi}}\left(t\right) \nonumber \\& - \int_{-\tau_1}^0 \boldsymbol{\dot{\Psi}}^T\left(t+\theta\right)Z_1\boldsymbol{\dot{\Psi}}\left(t+\theta\right)d\theta \nonumber \\& - \int_{-\tau_2}^0 \boldsymbol{\dot{\Psi}}^T\left(t+\theta\right)Z_2\boldsymbol{\dot{\Psi}}\left(t+\theta\right)d\theta \nonumber \\ & - \int_{-\tau_2}^{-\tau_1} \boldsymbol{\dot{\Psi}}^T\left(t+\theta\right)Z_3\boldsymbol{\dot{\Psi}}\left(t+\theta\right)d\theta
\end{align}
Let,
\begin{equation}\label{eqn_l11b}
\boldsymbol{\widehat{X}} = \left[\begin{matrix}\boldsymbol{X}\left(t\right)^T & \boldsymbol{X}\left(t-\tau_{1}\right)^T & \boldsymbol{X}\left(t-\tau_{2}\right)^T\end{matrix} \right]
\end{equation}
then,
\begin{equation}\label{eqn_l12}
\begin{split}
\boldsymbol{\dot{\Psi}} &= \mathcal{E}\boldsymbol{\dot{X}} \\
&=\mathcal{E}\left[\begin{matrix}\mathcal{A}_0 & \mathcal{A}_1 & \mathcal{A}_2\end{matrix} \right]\boldsymbol{\widehat{X}}^T
\end{split}
\end{equation}
\begin{equation}\label{eqn_l13}
\boldsymbol{\dot{\Psi}}^T\left(t\right)P\boldsymbol{\Psi}\left(t\right) = \boldsymbol{\widehat{X}} \left[\begin{matrix}\mathcal{A}_0 & \mathcal{A}_1 & \mathcal{A}_2\end{matrix} \right]^T \mathcal{E}^T P \mathcal{E}\boldsymbol{X}\left(t\right)
\end{equation}
\begin{equation}\label{eqn_l14}
\boldsymbol{\Psi}^T\left(t\right)P\boldsymbol{\dot{\Psi}}\left(t\right) = \boldsymbol{X}\left(t\right)^T\mathcal{E}^T P \mathcal{E}\left[\begin{matrix}\mathcal{A}_0 & \mathcal{A}_1 & \mathcal{A}_2\end{matrix} \right]\boldsymbol{\widehat{X}}^T
\end{equation}
\begin{equation}\label{eqn_l15}
\boldsymbol{\Psi}^T\left(t\right)Q_1\boldsymbol{\Psi}\left(t\right) = \boldsymbol{X}\left(t\right)^T\mathcal{E}^T Q_1 \mathcal{E}\boldsymbol{X}\left(t\right)
\end{equation}
\begin{equation}\label{eqn_l16}
\boldsymbol{\Psi}^T\left(t\right)Q_2\boldsymbol{\Psi}\left(t\right) = \boldsymbol{X}\left(t\right)^T\mathcal{E}^T Q_2 \mathcal{E}\boldsymbol{X}\left(t\right)
\end{equation}
\begin{equation}\label{eqn_l17}
\begin{split}
\boldsymbol{\Psi}^T\left(t-\tau_{1}\right)Q_1\boldsymbol{\Psi} & \left(t-\tau_{1}\right) = \\& \boldsymbol{X}\left(t-\tau_1\right)^T\mathcal{E}^T Q_1 \mathcal{E}\boldsymbol{X}\left(t-\tau_1\right)
\end{split}
\end{equation}
\begin{equation}\label{eqn_l18}
\begin{split}
\boldsymbol{\Psi}^T\left(t-\tau_{2}\right)Q_2\boldsymbol{\Psi} & \left(t-\tau_{2}\right) = \\& \boldsymbol{X}\left(t-\tau_2\right)^T\mathcal{E}^T Q_2 \mathcal{E}\boldsymbol{X}\left(t-\tau_2\right)
\end{split}
\end{equation}
For $i=\{1,2,3\}$,
\begin{equation}\label{eqn_l19}
\begin{split}
\boldsymbol{\dot{\Psi}}^T & \left(t\right)Z_i\boldsymbol{\dot{\Psi}}\left(t\right) = \\& \boldsymbol{\widehat{X}} \left[\begin{matrix}\mathcal{A}_0 & \mathcal{A}_1 & \mathcal{A}_2\end{matrix} \right]^T \mathcal{E}^T Z_i\mathcal{E} \left[\begin{matrix}\mathcal{A}_0 & \mathcal{A}_1 & \mathcal{A}_2\end{matrix} \right]\boldsymbol{\widehat{X}}^T 
\end{split}
\end{equation}
Consider a set of matrices,
\begin{equation}\label[ineq]{eqn_l20}
\left[\begin{matrix} H_{11} & H_{12} & H_{13} & H_{14}\\ H_{12}^T & H_{22} & H_{23} & H_{24}\\ H_{13}^T & H_{23}^T & H_{33}& H_{34} \\ H_{14}^T & H_{24}^T & H_{34}^T & Z_{1} \end{matrix}\right] \succcurlyeq 0
\end{equation}
\begin{equation}\label[ineq]{eqn_l21}
\left[\begin{matrix} I_{11} & I_{12} & I_{13} & I_{14}\\ I_{12}^T & I_{22} & I_{23} & I_{24}\\ I_{13}^T & I_{23}^T & I_{33}& I_{34} \\ I_{14}^T & I_{24}^T & I_{34}^T & Z_{2} \end{matrix}\right] \succcurlyeq 0
\end{equation}
\begin{equation}\label[ineq]{eqn_l22}
\left[\begin{matrix} J_{11} & J_{12} & J_{13} & J_{14}\\ J_{12}^T & J_{22} & J_{23} & J_{24}\\ J_{13}^T & J_{23}^T & J_{33}& J_{34} \\ J_{14}^T & J_{24}^T & J_{34}^T & Z_{3} \end{matrix}\right] \succcurlyeq 0
\end{equation}
Where $H_{ij}, I_{ij}, J_{ij}\; \forall i \in [1,3]$ $\forall j \in [1,4]$ are some arbitrary matrices to be found by an LMI solver with size $2n \times 2n$. 
\\
Let,
\begin{equation}\label{eqn_l22b}
\boldsymbol{\widehat{\Psi}}_{\dot{\theta}}=\left[\begin{matrix} \boldsymbol{\Psi}\left(t\right)^T & \boldsymbol{\Psi}\left(t-\tau_1\right)^T & \boldsymbol{\Psi}\left(t-\tau_2\right)^T& \boldsymbol{\dot{\Psi}}\left(t+\theta\right)^T\end{matrix}\right]
\end{equation}
then,
\begin{equation}\label[ineq]{eqn_l23}
\int_{-\tau_1}^0 \boldsymbol{\widehat{\Psi}}_{\dot{\theta}} \left[\begin{matrix} H_{11} & H_{12} & H_{13} & H_{14}\\ H_{12}^T & H_{22} & H_{23} & H_{24}\\ H_{13}^T & H_{23}^T & H_{33}& H_{34} \\ H_{14}^T & H_{24}^T & H_{34}^T & Z_{1} \end{matrix}\right] \boldsymbol{\widehat{\Psi}}_{\dot{\theta}}^T d\theta \geq 0
\end{equation}
\begin{equation}\label[ineq]{eqn_l24}
\int_{-\tau_2}^0 \boldsymbol{\widehat{\Psi}}_{\dot{\theta}} \left[\begin{matrix} I_{11} & I_{12} & I_{13} & I_{14}\\ I_{12}^T & I_{22} & I_{23} & I_{24}\\ I_{13}^T & I_{23}^T & I_{33}& I_{34} \\ I_{14}^T & I_{24}^T & I_{34}^T & Z_{2} \end{matrix}\right] \boldsymbol{\widehat{\Psi}}_{\dot{\theta}}^T d\theta \geq 0
\end{equation}
\begin{equation}\label[ineq]{eqn_l25}
\int_{-\tau_2}^{-\tau_1} \boldsymbol{\widehat{\Psi}}_{\dot{\theta}} \left[\begin{matrix} J_{11} & J_{12} & J_{13} & J_{14}\\ J_{12}^T & J_{22} & J_{23} & J_{24}\\ J_{13}^T & J_{23}^T & J_{33}& J_{34} \\ J_{14}^T & J_{24}^T & J_{34}^T & Z_{3} \end{matrix}\right] \boldsymbol{\widehat{\Psi}}_{\dot{\theta}}^T d\theta \geq 0
\end{equation}
The matrices in \cref{eqn_l20,eqn_l21,eqn_l22} are chosen to satisfy expression in \cref{eqn_l23,eqn_l24,eqn_l25}, which further simplify $\dot{V}$ in \cref{eqn_l11}. Parts of \cref{eqn_l11} consisting integrals with multiplication two variable in terms of $\theta$ are eliminated when added with \cref{eqn_l23,eqn_l24,eqn_l25}, since $\dot{V} + \{positive\,semidefinite\} < 0$ implies $\dot{V}<0$. Substituting \cref{eqn_l11b,eqn_l12,eqn_l13,eqn_l14,eqn_l15,eqn_l16,eqn_l17,eqn_l18,eqn_l19} in \cref{eqn_l11}, adding \cref{eqn_l23,eqn_l24,eqn_l25} and further solving leftover integrals, \cref{eqn_l26} is obtained. 
\begin{equation}\label[ineq]{eqn_l26}
\dot{V} \leq  \boldsymbol{\widehat{X}}\left[\begin{matrix} G_{11} & G_{12} & G_{13} \\ G_{12}^T & G_{22} & G_{23} \\ G_{13}^T & G_{23}^T & G_{33} \end{matrix}\right]\boldsymbol{\widehat{X}}^T
\end{equation}
where,
\begin{align}G_{11}=&\mathcal{A}_0^T\mathcal{E}^TP\mathcal{E} + \mathcal{E}^TP\mathcal{E}\mathcal{A}_0 + \mathcal{E}^TQ_{1}\mathcal{E} + \mathcal{E}^TQ_{2}\mathcal{E} + \\&  \mathcal{A}_0^T\mathcal{E}^T\Xi\mathcal{E}\mathcal{A}_0  + \tau_1\mathcal{E}^TH_{11}\mathcal{E} +  \mathcal{E}^TH_{14}\mathcal{E} + \mathcal{E}^TH_{14}^T\mathcal{E} +  \nonumber \\& \tau_2\mathcal{E}^TI_{11}\mathcal{E} + \mathcal{E}^TI_{14}\mathcal{E} + \mathcal{E}^TI_{14}^T\mathcal{E} + \left(\tau_2 - \tau_1\right)\mathcal{E}^TJ_{11}\mathcal{E} \nonumber \\
G_{12}=&\mathcal{E}^TP\mathcal{E}\mathcal{A}_1 + \mathcal{A}_0^T\mathcal{E}^T\Xi\mathcal{E}\mathcal{A}_1 + \tau_1\mathcal{E}^TH_{12}\mathcal{E} - \mathcal{E}^TH_{14}\mathcal{E} \nonumber \\& + \mathcal{E}^TH_{24}^T\mathcal{E}  + \tau_2\mathcal{E}^TI_{12}\mathcal{E} + \mathcal{E}^TI_{24}^T\mathcal{E} + \nonumber \\& \left(\tau_2-\tau_1\right)\mathcal{E}^TJ_{12}\mathcal{E} + \mathcal{E}^TJ_{14}\mathcal{E} \\
G_{13}=&\mathcal{E}^TP\mathcal{E}\mathcal{A}_2 + \mathcal{A}_0^T\mathcal{E}^T\Xi\mathcal{E}\mathcal{A}_2 + \tau_1\mathcal{E}^TH_{13}\mathcal{E} + \nonumber \\& \mathcal{E}^TH_{34}^T\mathcal{E} + \tau_2\mathcal{E}^TI_{13}\mathcal{E}  + \mathcal{E}^TI_{34}^T\mathcal{E} - \mathcal{E}^TI_{14}\mathcal{E} + \nonumber \\& \left(\tau_2-\tau_1\right)\mathcal{E}^TJ_{13}\mathcal{E} - \mathcal{E}^TJ_{14}\mathcal{E} \\
G_{22}=&-\mathcal{E}^TQ_{1}\mathcal{E} + \mathcal{A}_1^T\mathcal{E}^T\Xi\mathcal{E}\mathcal{A}_1 +  \tau_1\mathcal{E}^TH_{22}\mathcal{E} - \mathcal{E}^TH_{24}\mathcal{E} \nonumber \\& -\mathcal{E}^TH_{24}^T\mathcal{E}  + \tau_2\mathcal{E}^TI_{22}\mathcal{E} + \left(\tau_2 - \tau_1\right)\mathcal{E}^TJ_{22}\mathcal{E} \nonumber \\& + \mathcal{E}^TJ_{24}\mathcal{E} + \mathcal{E}^TJ_{24}^T\mathcal{E} \\ 
G_{23}=&\mathcal{A}_1^T\mathcal{E}^T\Xi\mathcal{E}\mathcal{A}_2 + \tau_1\mathcal{E}^TH_{23}\mathcal{E} - \mathcal{E}^TH_{34}^T\mathcal{E}  \nonumber \\& + \tau_2\mathcal{E}^TI_{23}\mathcal{E} - \mathcal{E}^TI_{24}\mathcal{E} +\left(\tau_2-\tau_1\right)\mathcal{E}^TJ_{23}\mathcal{E} \nonumber \\& - \mathcal{E}^TJ_{24}\mathcal{E} + \mathcal{E}^TJ_{34}^T\mathcal{E} \\
G_{33}=&-\mathcal{E}^TQ_{2}\mathcal{E} + \mathcal{A}_2^T\mathcal{E}^T\Xi\mathcal{E}\mathcal{A}_2 + \tau_1\mathcal{E}^TH_{33}\mathcal{E} - \mathcal{E}^TI_{34}\mathcal{E} \nonumber \\&  -\mathcal{E}^TI_{34}^T\mathcal{E} + \tau_2\mathcal{E}^TI_{33}\mathcal{E} + \left(\tau_2 - \tau_1\right)\mathcal{E}^TJ_{33}\mathcal{E} \nonumber \\& - \mathcal{E}^TJ_{34}\mathcal{E} - \mathcal{E}^TJ_{34}^T\mathcal{E} \\
\Xi = &\tau_1Z_1 + \tau_2Z_2 + \left(\tau_2-\tau_1\right)Z_3
\end{align}

The matrices $P$, $Q_1$, $Q_2$, $Z_i$, $H_{ij}$, $H_{ij}^T$, $I_{ij}$, $I_{ij}^T$, $J_{ij}$, $J_{ij}^T$ and $\mathcal{F}_i$ from \cref{lem2} $\forall i \in [1,3]$, $\forall j \in [1,4]$ will generate corresponding $\widetilde{P}$, $\widetilde{Q}_1$, $\widetilde{Q}_2$, $\widetilde{Z}_i$, $\widetilde{H}_{ij}$, $\widetilde{H}_{ij}^T$, $\widetilde{I}_{ij}$, $\widetilde{I}_{ij}^T$, $\widetilde{J}_{ij}$, $\widetilde{J}_{ij}^T$ and $\widetilde{\mathcal{F}}_i$ $\forall i \in [1,3]$, $\forall j \in [1,4]$ of size $(2n-2) \times (2n-2)$ when multiplied with eigenvector matrices $U^T$, $U$ at appropriate positions. The corresponding $G_{ij}$ is as given in \cref{eqn_l27}. \\
\begin{equation}\label{eqn_l27}
\begin{split}
G_{ij}=U\left[\begin{matrix} \widetilde{G}_{ij_{(2n-2) \times (2n-2)}} & 0_{_{(2n-2) \times 2}} \\
0_{_{2 \times (2n-2)}} & 0_{_{2 \times 2}}\end{matrix}\right]U^T &\\
\{i,j\} \in [1,3];\; \widetilde{G}_{ji} = \widetilde{G}_{ji}^T & \; \forall i \neq j
\end{split}
\end{equation}
With above set of reduced order matrices, the LMIs given in \cref{eqn_l6,eqn_l7,eqn_l8,eqn_l9} can be obtained.
\QEDB
\end{proof}
Feasibility of LMIS in \cref{eqn_l6,eqn_l7,eqn_l8,eqn_l9} determine the consensus reachability of the multi-agent system affected by time-delays. They can be solved by using solvers like SeDuMi \cite{sedumi}, Matlab LMI Lab solver etc..
\subsection{Nyquist stability approach}\label{freq}
Stability analysis of the linear element for different control inputs is performed using frequency domain analysis and Nyquist stability criterion which is discussed in \cref{case-1,case-2}.
\subsubsection{First control law}\label{case-1}
Consider a control input $u_{i1}$ with input delay $\tau_1$ and communication delay $\tau_2$ as given in \cref{eqn1c} assuming $\tau_1 \leq \tau_2$.
%\begin{equation}\label{eqn_f2}
%\begin{split}
%u_{i1}\left(t\right) = -v_i\left(t-\tau_1\right)+\frac{1}{\sum_{j=1}^n a_{ij}} \sum_{j=1}^n \big[ a_{ij}\left(x_j\left(t-\tau_2\right) - x_i\left(t-\tau_1\right)\right)\big]
%\end{split}
%\end{equation}
\begin{theorem}\label{th_freq1}
The system represented by \cref{eqn1b,eqn1c} is stable if and only if,
\begin{equation}\label{eqn_f1}
\dfrac{|\tilde{\lambda}_k|}{\sqrt{\overline{\omega}^4-2\overline{\omega}^3\sin\left(\overline{\omega} \tau_1\right)+\overline{\omega}^2\left(1-2\cos\left(\overline{\omega} \tau_1\right)\right)+1}} < 1
\end{equation}
where $\overline{\omega}$ satisfies 
\begin{equation}\label{eqn_f2}
\begin{split}
-\pi= &-\overline{\omega} \tau_2 + \arg\left(-\tilde{\lambda}_k\right) \\& - \arctan\left(\dfrac{\overline{\omega}\cos\left(\overline{\omega} \tau_1\right)-\sin\left(\overline{\omega} \tau_1\right)}{-\overline{\omega}^2+\cos\left(\overline{\omega}\tau_1\right)+ \overline{\omega} \sin\left(\overline{\omega} \tau_1\right)}\right)
\end{split}
\end{equation}
\end{theorem}
%\vspace{-1.09em}
\begin{proof}
\normalfont
The system in \cref{eqn1b,eqn1c} can also be represented as,
\begin{equation}\label{eqn_f3}
\begin{split}
\ddot{x}_i \left(t\right) =& -\dot{x}_i\left(t-\tau_1 \right) + \\& \frac{1}{\sum_{j=1}^n a_{ij}} \sum_{j=1}^n a_{ij}\left(x_j\left(t-\tau_2\right) - x_i\left(t-\tau_1\right)\right)
\end{split}
\end{equation}
Converting it into $s$ domain will give the characteristic expression as,
\begin{equation}\label{eqn_f4}
\begin{split}
& \left|s^2I + \left(s+1\right)Ie^{-s\tau_1}- \tilde{\mathcal{A}}e^{-s\tau_2}\right| = \\& \prod_{k=1}^{n}\left(s^2I + \left(s+1\right)Ie^{-s\tau_1}- \tilde{\lambda}_k e^{-s\tau_2}\right)
\end{split}
\end{equation}
Consider $\forall \tilde{\lambda}_k \neq 0$,
\begin{equation}\label{eqn_f5}
G_k\left(s\right) = \dfrac{-\tilde{\lambda}_k e^{-s\tau_2}}{s^2+\left(s+1\right)e^{-s\tau_1}}
\end{equation}
Magnitude expression of \cref{eqn_f5} is given by,
\begin{equation}\label{eqn_f6}
\begin{split}
& \big|G_k\left(j\omega\right)\big| = \\& \dfrac{|\tilde{\lambda}_k|}{\sqrt{\omega^4-2\omega^3\sin\left(\omega \tau_1\right)+\omega^2\left(1-2\cos\left(\omega \tau_1\right)\right)+1}}
\end{split}
\end{equation}
Phase expression of \cref{eqn_f5} is given by,
\begin{equation}\label{eqn_f7}
\begin{split}
\phase{ G_k\left(j\omega\right) }= & -\omega \tau_2 + \arg\left(-\tilde{\lambda}_k\right) - \\& \arctan\left(\dfrac{\omega\cos\left(\omega \tau_1\right)-\sin\left(\omega \tau_1\right)}{-\omega^2+\cos\left(\omega\tau_1\right)+\omega\sin\left(\omega\tau_1\right)}\right)
\end{split}
\end{equation}
Let us assume at $\omega = \overline{\omega}$, the Nyquist plot intersects with negative real axis. The phase at $\omega = \overline{\omega}$ is $-\pi$,
\begin{equation}\label{eqn_f8}
\begin{split}
-\pi= & -\overline{\omega} \tau_2 + \arg\left(-\tilde{\lambda}_k\right) - \\& \arctan\left(\dfrac{\overline{\omega}\cos\left(\overline{\omega} \tau_1\right)-\sin\left(\overline{\omega} \tau_1\right)}{-\overline{\omega}^2+\cos\left(\overline{\omega}\tau_1\right)+ \overline{\omega} \sin\left(\overline{\omega} \tau_1\right)}\right)
\end{split}
\end{equation}
By applying Nyquist stability criterion, magnitude given in \cref{eqn_f6} should satisfy the condition as given in \cref{eqn_f9}.
\begin{equation}\label[ineq]{eqn_f9}
\dfrac{|\tilde{\lambda}_k|}{\sqrt{\overline{\omega}^4-2\overline{\omega}^3\sin\left(\overline{\omega} \tau_1\right)+\overline{\omega}^2\left(1-2\cos\left(\overline{\omega} \tau_1\right)\right)+1}} < 1
\end{equation}
\QEDB
\end{proof}
%\Cref{eqn10a,eqn10b} are used to calculate the stable ranges of $\tau_1,\tau_2$ with input $u_{i1}$.
\subsubsection{Second control law}\label{case-2}
Consider control input $u_{i2}$ given in \cref{eqn1d}, where,
\begin{theorem}\label{th_freq2}
The system represented in \cref{eqn1b,eqn1d} is stable if and only if,
\begin{equation}\label{eqn_f10}
\dfrac{|\tilde{\lambda}_k|\sqrt{1+\overline{\omega}^2}}{\sqrt{\overline{\omega}^4-2\overline{\omega}^3\sin\left(\overline{\omega} \tau_1\right)+\overline{\omega}^2\left(1-2\cos\left(\overline{\omega} \tau_1\right)\right)+1}} < 1
\end{equation}
where $\overline{\omega}$ satisfies,
\begin{equation}\label{eqn_f11}
\begin{split}
-\pi = - & \overline{\omega} \tau_2 + \arg\left(-\tilde{\lambda}_k\right) +   \arctan\left(\overline{\omega}\right)- \\& \arctan\left(\dfrac{\overline{\omega}\cos\left(\overline{\omega} \tau_1\right)-\sin\left(\overline{\omega} \tau_1\right)}{-\overline{\omega}^2+\cos\left(\overline{\omega}\tau_1\right)+\overline{\omega} \sin\left(\overline{\omega}\tau_1\right)}\right) 
\end{split}
\end{equation}
\end{theorem}
\begin{proof}
\normalfont
The proof follows a similar procedure as given in \cref{th_freq1}
\QEDB
\end{proof}
\section{Simulation and Implementation Results}\label{sec4}
The communication topologies considered for simulation and implementation are depicted in \cref{graph1,graph2}. The graph in \cref{graph1} is undirected, strongly connected and 2-regular balanced with each node receiving states' information from two neighbours and sending states' information to the same neighbours. The graph in {\cref{graph2}} is directed, has a spanning tree and unbalanced. Using the results obtained in \cref{th_lyp,th_freq1,th_freq2}, limits on communication time-delay for given input-delays are calculated for both the topologies. The feasibility of LMIs given in \cref{th_lyp} is solved using SeDuMi \cite{sedumi} solver for Matlab/Octave. The expressions in \cref{th_freq1,th_freq2} have three unknowns $(\omega, \tau_1, \tau_2)$, a unique solution can be obtained if it is assumed that $\tau_1 = \tau_2$ or else, stable range of $\tau_2$ for a given $\tau_1$ have to be found. Dominant pole for both the topologies in \cref{graph1,graph2} is $-1$, Nyquist plot used in one of the cases with $\lambda = -1$ and assumption $\tau_1 = 0.4$ is shown in \cref{nyq1}. It can be observed that the system is stable if $\tau_2 < 0.48$ (\cref{nyq1}). The time-delay tolerances are tested with simulations and on a hardware setup (\cref{hw_fig}). The results obtained using Lyapunov-Krasovskii and Nyquist approaches from \cref{th_lyp,th_freq1,th_freq2} are tabulated in \cref{tbl1,tbl2}, corresponding plots of $\tau_1$ vs $\tau_2$ depicting stable regions are given in \cref{t1_vs_t2_a,t1_vs_t2_b}.

\begin{figure}[!h]
\centering
\begin{subfigure}[b]{0.45\linewidth}
\centering
\includegraphics[scale=1]{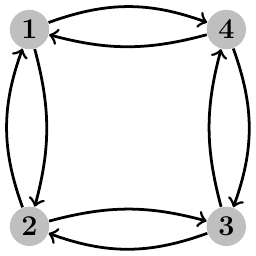}
\caption{Four agents.\label{graph1}}
\end{subfigure}\qquad
\begin{subfigure}[b]{0.45\linewidth}
\centering
\includegraphics[scale=1]{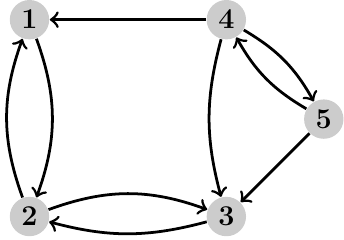}
\caption{Five agents.\label{graph2}}
\end{subfigure}
\caption{Graphs of communication topologies.}
\end{figure}

%\vspace{-0.5cm}

\begin{figure}[!h]
\centering
\includegraphics[scale=0.14]{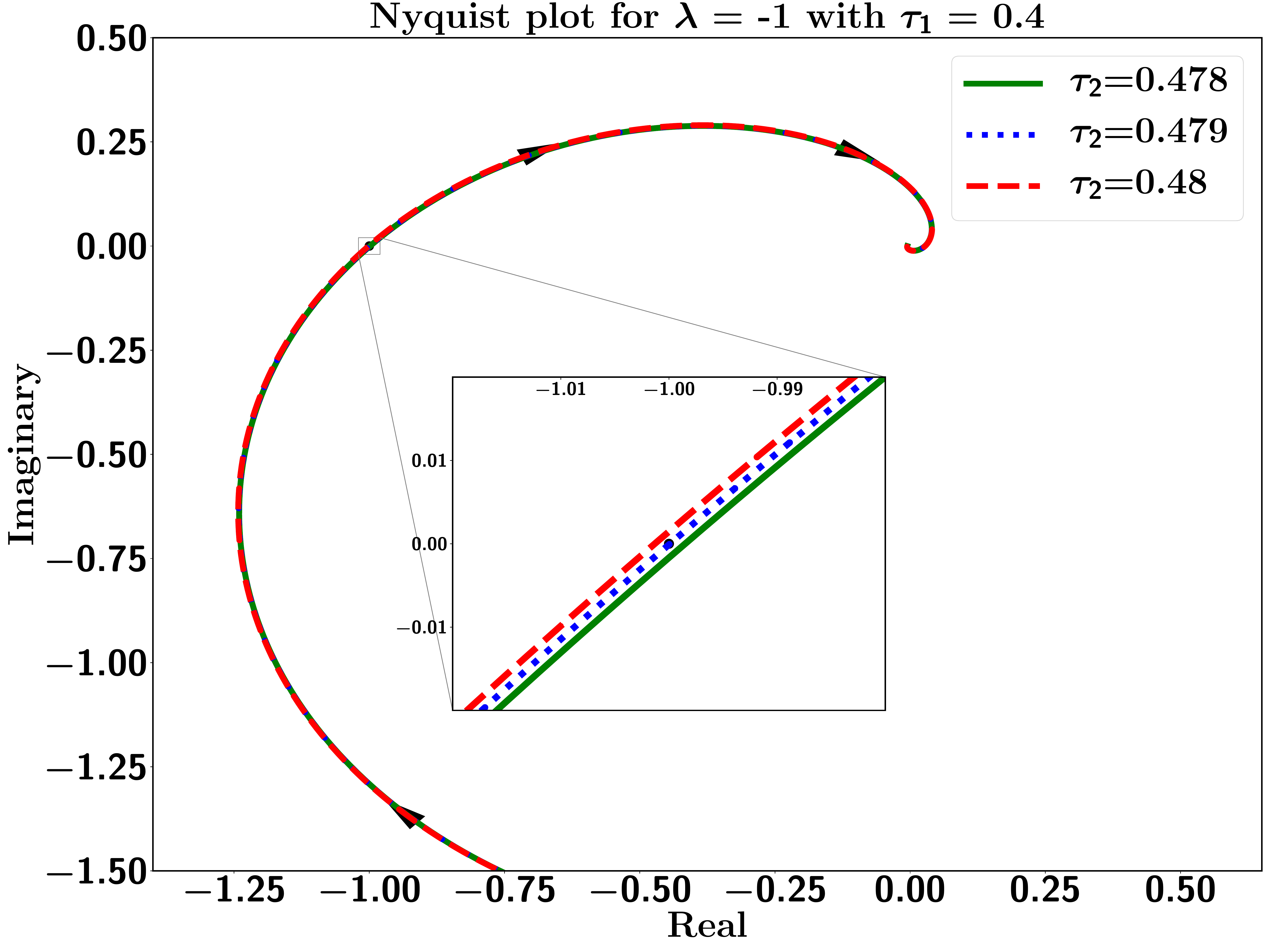}
\caption{Nyquist plot with $u_{i1}$ in \cref{eqn1c} for different time delays.\label{nyq1}}
\end{figure}

%\vspace{-0.5cm}

\begin{figure}[!h]
\centering
\includegraphics[scale=0.09]{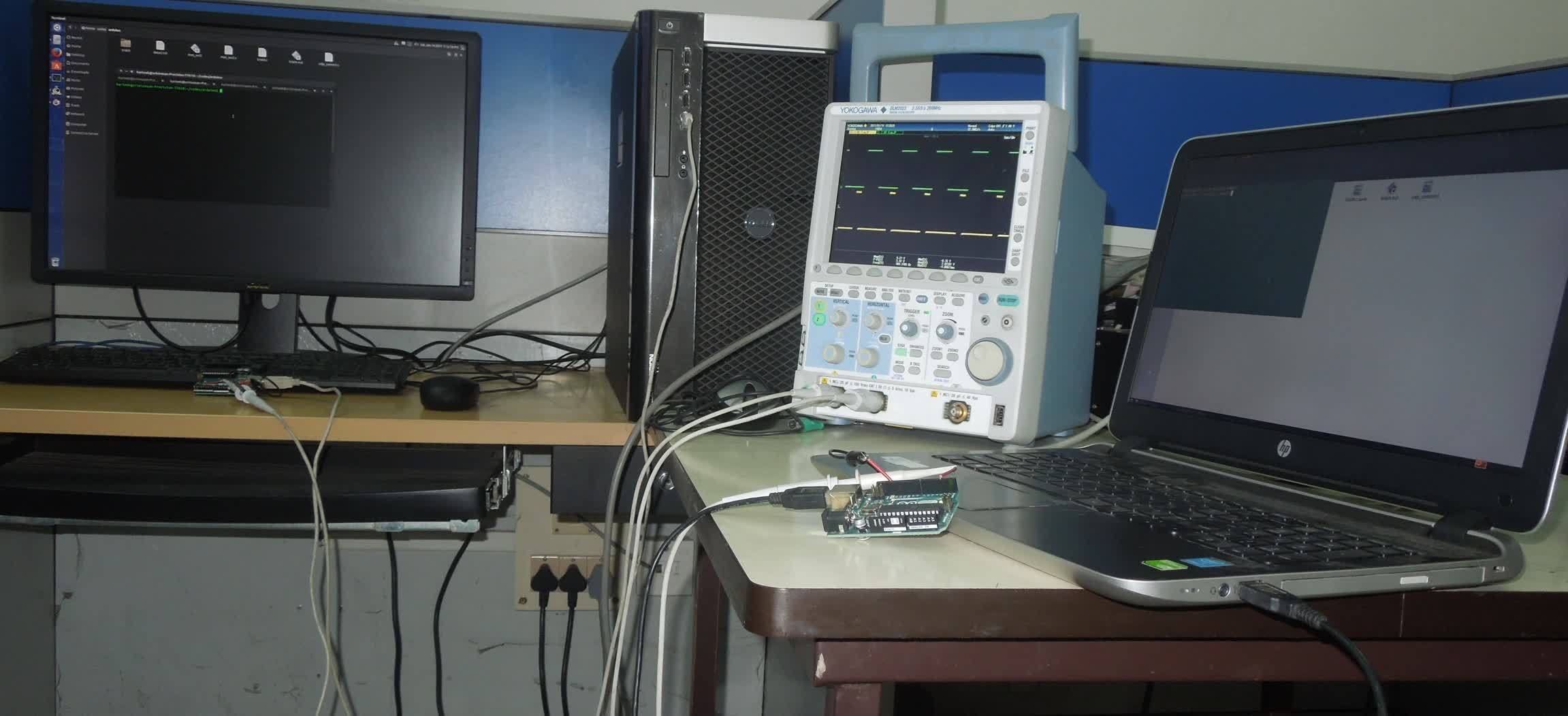}
\caption{Hardware setup connected by LAN.\label{hw_fig}}
\end{figure}

%\vspace{-0.5cm}

From the results in \cref{tbl1,tbl2} and \cref{t1_vs_t2_a,t1_vs_t2_b}, it can be deduced that the Lyapunov-Krasovskii approach is conservative compared to Nyquist approach with respect to time-delay. Conservativeness of Lyapunov-Krasovskii approach is more evident for topology in \cref{graph2} with control law in \cref{eqn1d}. Multi-agent systems connected by topologies in \cref{graph1,graph2} reach consensus with full range of time-delay given by Nyquist approach for control laws in \cref{eqn1c,eqn1d}. Nyquist approach in \cref{freq} is not applicable to control laws given in \cref{eqn1e,eqn1f}. Solving LMIs is computationally more intensive compared to solving of equations from \cref{th_freq1,th_freq2}. Moreover, the increase in computational time of solving LMIs is exponentially as the number of nodes are increased and the increase with Nyquist approach is linear.
 
\begin{figure}[!h]
\centering
\begin{subfigure}[b]{\linewidth}
\centering
\includegraphics[scale=0.14]{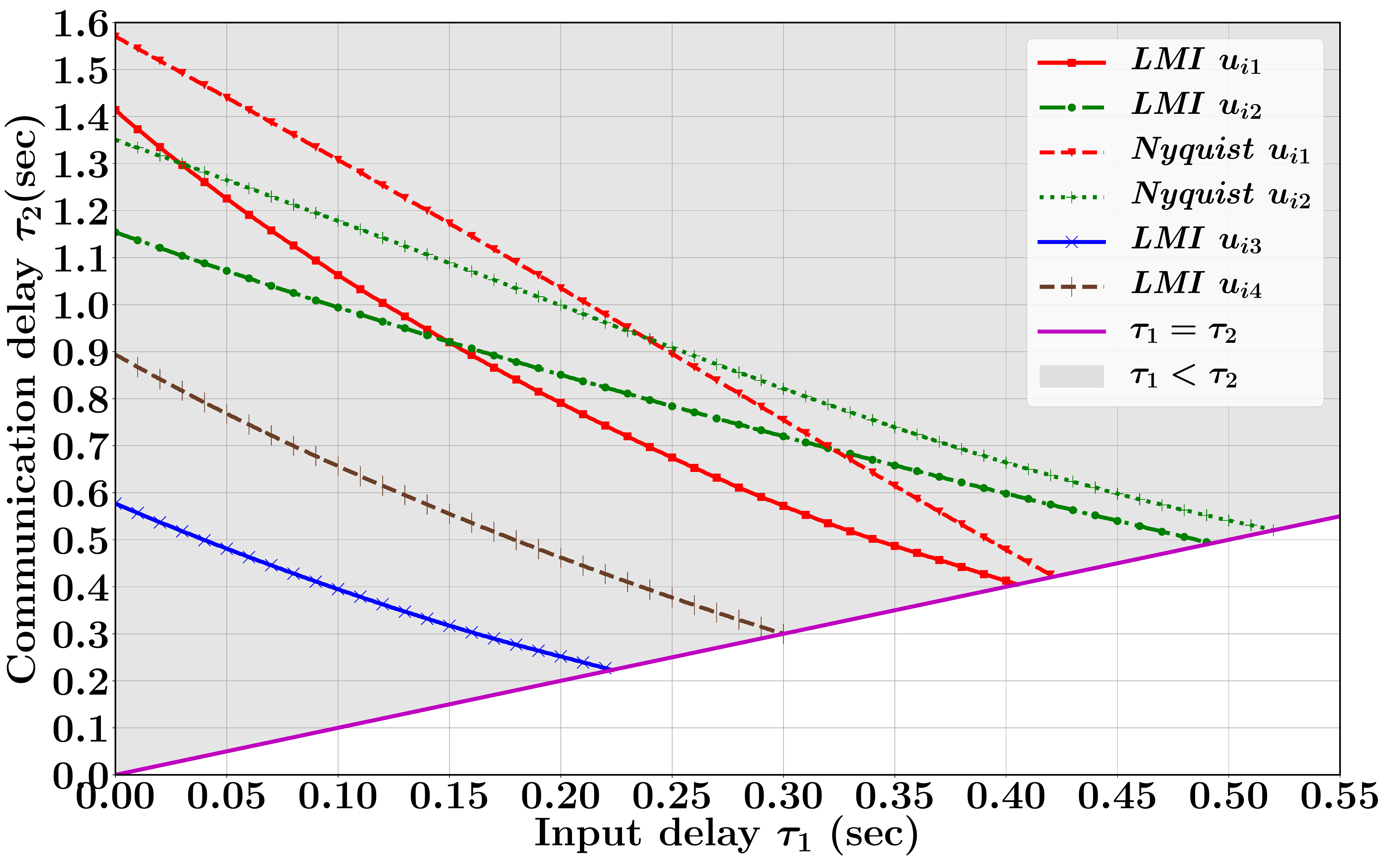}
\caption{Plot of $\tau_1$ vs $\tau_2$ for topology in \cref{graph1}.\label{t1_vs_t2_a}}
\end{subfigure}
\begin{subfigure}[b]{\linewidth}
\centering
\includegraphics[scale=0.14]{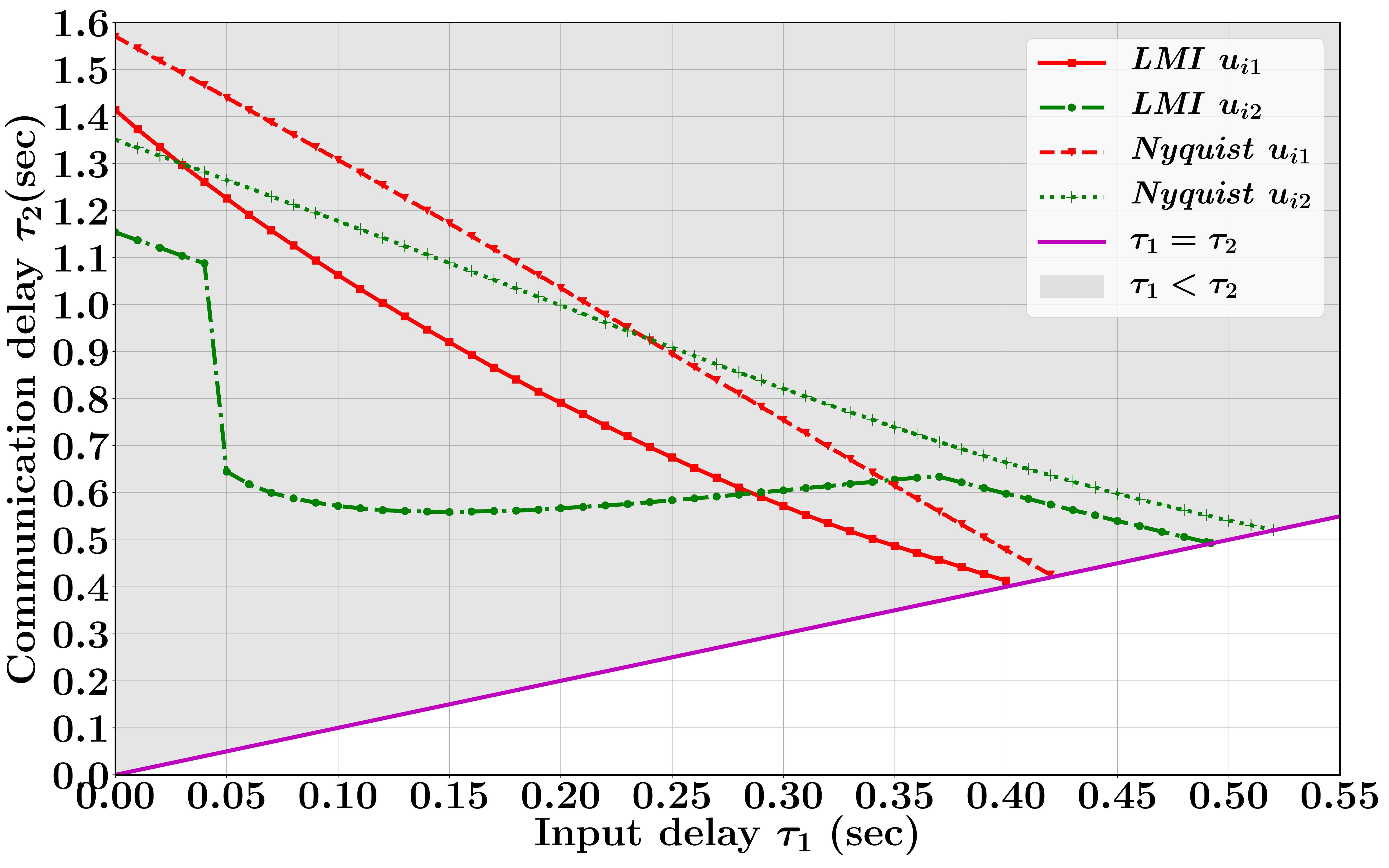}
\caption{Plot of $\tau_1$ vs $\tau_2$ for topology in \cref{graph2}.\label{t1_vs_t2_b}}
\end{subfigure}
\caption{Plots representing stable regions with $\tau_1 \leq \tau_2$.}
\end{figure}

%\vspace{-0.5cm}

\begin{table}[!h]
\processtable{Maximum value of $\tau_2$ for a given $\tau_1 \; (\tau_1 \leq \tau_2)$ with topology in \cref{graph1} and control inputs given in \cref{eqn1c,eqn1d,eqn1e,eqn1f}.\label{tbl1}}
{\begin{tabular*}{20pc}{@{\extracolsep{\fill}}lllllll@{}}\toprule
  \textbf{\boldmath$\tau_1$ (sec)} & \multicolumn{6}{c}{\textbf{\boldmath$\tau_2$ (sec)}} \\\midrule
    & \multicolumn{4}{c}{\emph{Lyapunov Approach}} & \multicolumn{2}{c}{\emph{Nyquist Approach}}\\
    %\cline{2-7}
    &{$u_{i1}$} &{$u_{i2}$} & {$u_{i3}$} &{$u_{i4}$} &{$u_{i1}$} &{$u_{i2}$} \\\midrule
    0 & 1.414 & 1.154  & 0.577 & 0.894 & 1.570 & 1.351  \\
    0.1 & 1.063 & 0.994  & 0.395 & 0.656 & 1.308 & 1.178  \\
    0.2 & 0.791 & 0.851  & 0.252 & 0.462 & 1.035 & 0.999  \\
    0.3 & 0.572 & 0.720  & --- & 0.30 & 0.755 & 0.821  \\
    0.4 & 0.413 & 0.598  & --- & --- & 0.479 & 0.664  \\
    0.5 & --- & ---  & --- & --- & --- & 0.541  \\
    0.6 & --- & ---  & --- & --- & --- & ---  \\
    \boldmath$\tau_1=\tau_2$ & 0.405 & 0.492  & 0.223 & 0.3 & 0.421 & 0.520  \\
    \botrule
   \end{tabular*}}{}
\end{table}

\begin{table}[!h]
\processtable{Maximum value of $\tau_2$ for a given $\tau_1 \; (\tau_1 \leq \tau_2)$ with topology in \cref{graph2} and control inputs given in \cref{eqn1c,eqn1d}.\label{tbl2}}
{\begin{tabular*}{20pc}{@{\extracolsep{\fill}}lllll@{}}\toprule
  \textbf{\boldmath$\tau_1$ (sec)} & \multicolumn{4}{c}{\textbf{\boldmath$\tau_2$ (sec)}} \\\midrule
    & \multicolumn{2}{c}{\emph{Lyapunov Approach}} & \multicolumn{2}{c}{\emph{Nyquist Approach}}\\
    %\cline{2-7}
    &{$u_{i1}$} &{$u_{i2}$} &{$u_{i1}$} &{$u_{i2}$} \\\midrule
    0 & 1.414 & 1.154 & 1.570 & 1.351  \\
    0.1 & 1.063 & 0.572 & 1.308 & 1.178  \\
    0.2 & 0.791 & 0.567 & 1.035 & 0.999  \\
    0.3 & 0.572 & 0.605 & 0.755 & 0.821  \\
    0.4 & 0.413 & 0.598 & 0.479 & 0.664  \\
    0.5 & --- & ---  & --- & 0.541  \\
    0.6 & --- & ---  & --- & ---  \\
    \boldmath$\tau_1=\tau_2$ & 0.405 & 0.492 & 0.421 & 0.520  \\
    \botrule
   \end{tabular*}}{}
\end{table}

Simulations are performed using scripts written in C to have uniformity with hardware implementation. Implementations of corresponding simulations are performed on a network of four agents. Four Arduino-Uno boards for topology in \cref{graph1} and five Arduino-Uno boards for topology in \cref{graph2} are considered as nodes of sensor networks, all of them are connected to host computers using serial interface. Host computers are connected by LAN switch locally and communication topology is based on the graphs shown in \cref{graph1,graph2}. All the implementations are performed after time synchronization at the start of each run with a local server through \emph{ntp} protocol. \emph{UDP} packet switching is used for communication and appropriate precautions like time-stamping of packets are taken to ensure that packets are received in order. Owing to limitations in capability of hardware, the step-size is chosen as $10ms$.  \cref{hw_fig} depicts hardware setup consisting of two agents with connection settings as described earlier, pulse width of both the agents can be observed in the display of two channel DSO.

\begin{figure}[!h]
\centering
\begin{subfigure}[b]{\linewidth}
\centering
\includegraphics[scale=0.14]{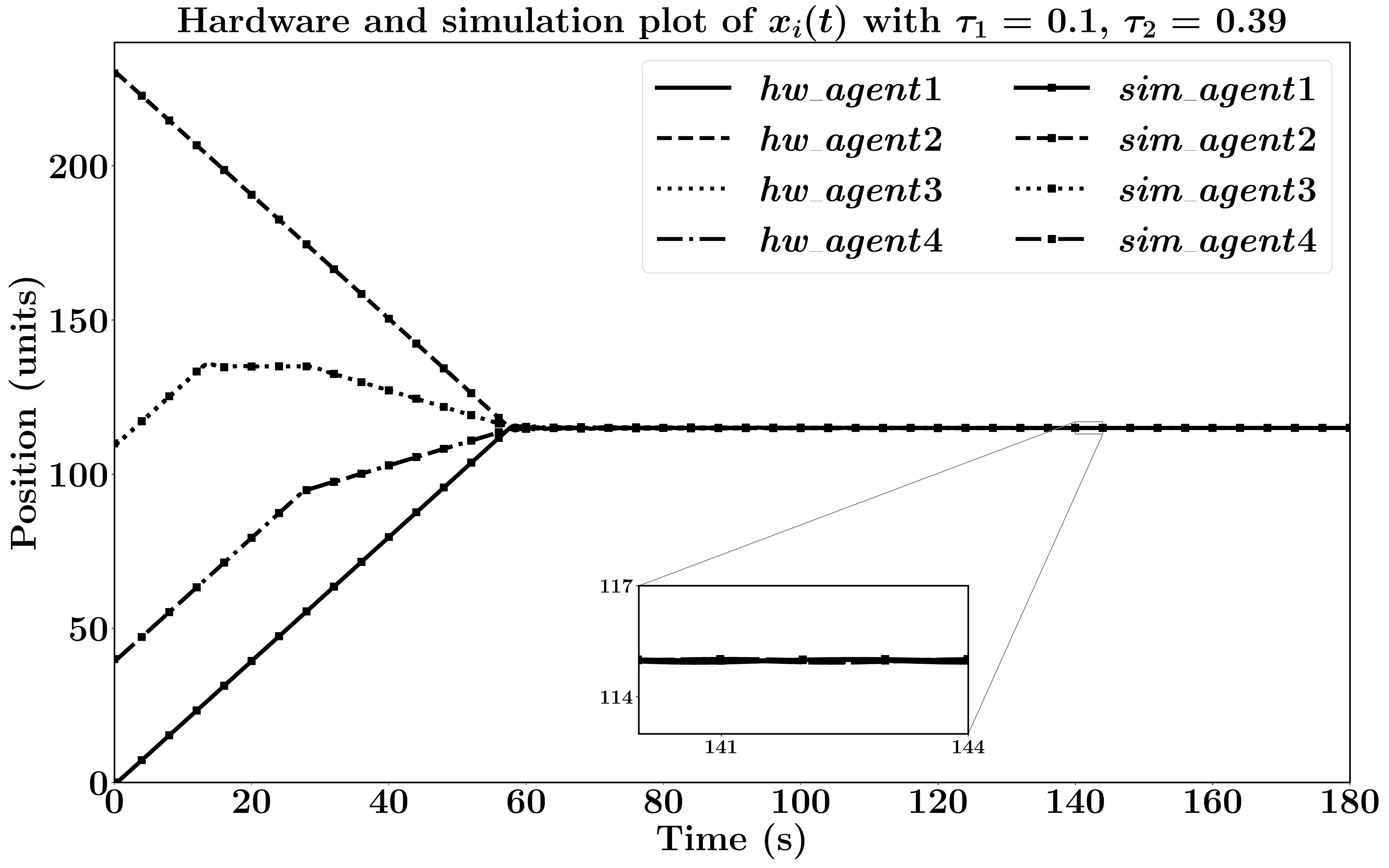}
%\caption{Plot with $u_{i3}$ in \cref{eqn1e} for non-uniform time delays.\label{res1p}}
\end{subfigure}
\begin{subfigure}[b]{\linewidth}
\centering
\includegraphics[scale=0.14]{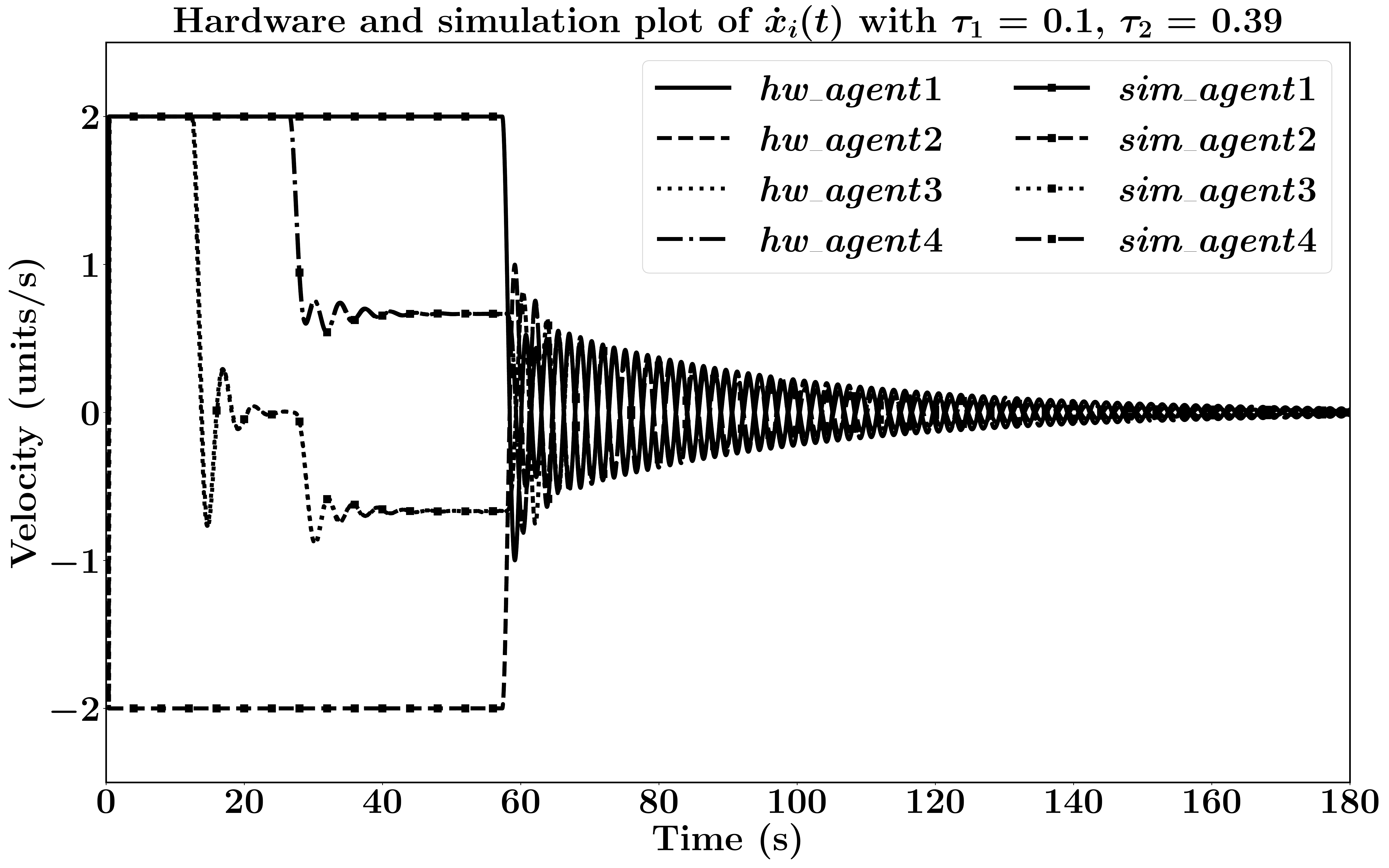}
%\caption{Plot with $u_{i3}$ in \cref{eqn1e} for non-uniform time delays.\label{res1v}}
\end{subfigure}
\caption{Plots with $u_{i3}$ in \cref{eqn1f} for asynchronous time delays.\label{res1}}
\end{figure}

Input time-delay $\tau_1$ and communication time-delay $\tau_2$ are user-defined and implemented from code. A minor deviation is visible between simulation and implementation results due to additional delay ($\approx 4ms$) in implementation due to actual processing and communication. The pulse width of PWM wave generated from an Arduino-Uno is considered as state $x_i$ and the rate of change is considered as $v_i$ of each agent. 

Initial conditions for both simulation and implementation for topology in \cref{graph1} are assumed to be, $\{x_{i0}\}=[0, 230, 110, 40]$ and $\{v_{i0}\}=[0,0,0,0]$. Plots with overlapping simulation and hardware results are given in \cref{res1,res2,res15,res1617} to show the effectiveness of theoretical results given in \cref{tbl1,tbl2,tbl3,t1_vs_t2_a,t1_vs_t2_b}. For control law in \cref{eqn1e} with input delay $\tau_1=0.1s$ and communication delay of $\tau_2=0.39s$, \cref{res1} depict states $x_i(t)$ and $\dot{x}_i(t)$ vs time in seconds respectively. Similarly, \cref{res2} show the plots of states for control law in \cref{eqn1f} with $\tau_1=0.2s$ and  of $\tau_2=0.46s$. It can be observed that the difference in states, $\boldsymbol{\Psi}$ is asymptotically converging to zero and reinforcing the effectiveness of the theoretical results given in \cref{tbl1,t1_vs_t2_a}.

\begin{figure}[!h]
\centering
\begin{subfigure}[b]{\linewidth}
\centering
\includegraphics[scale=0.14]{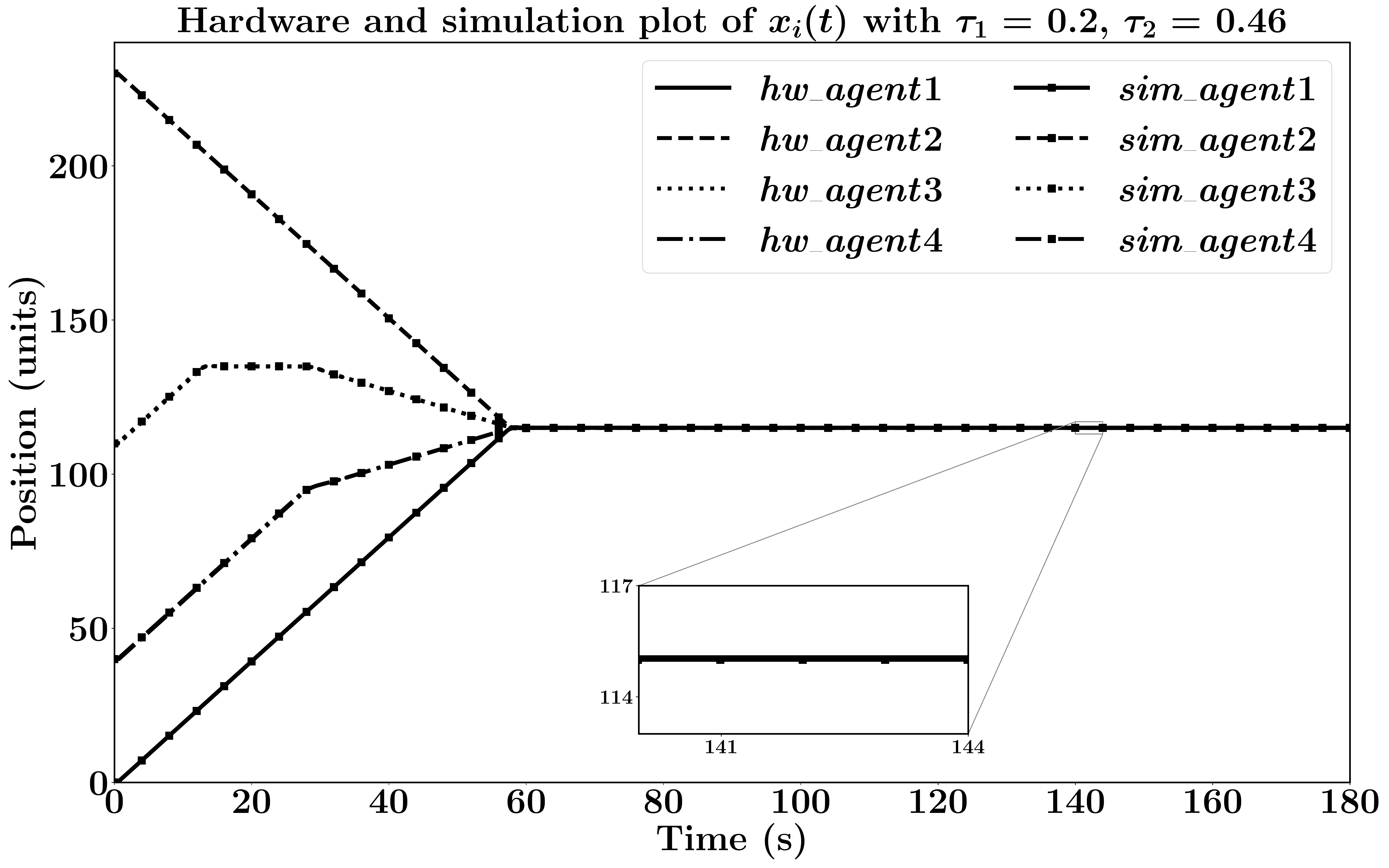}
%\caption{Plot with $u_{i4}$ in \cref{eqn1f} for non-uniform time delays.\label{res2p}}
\end{subfigure}
\begin{subfigure}[b]{\linewidth}
\centering
\includegraphics[scale=0.14]{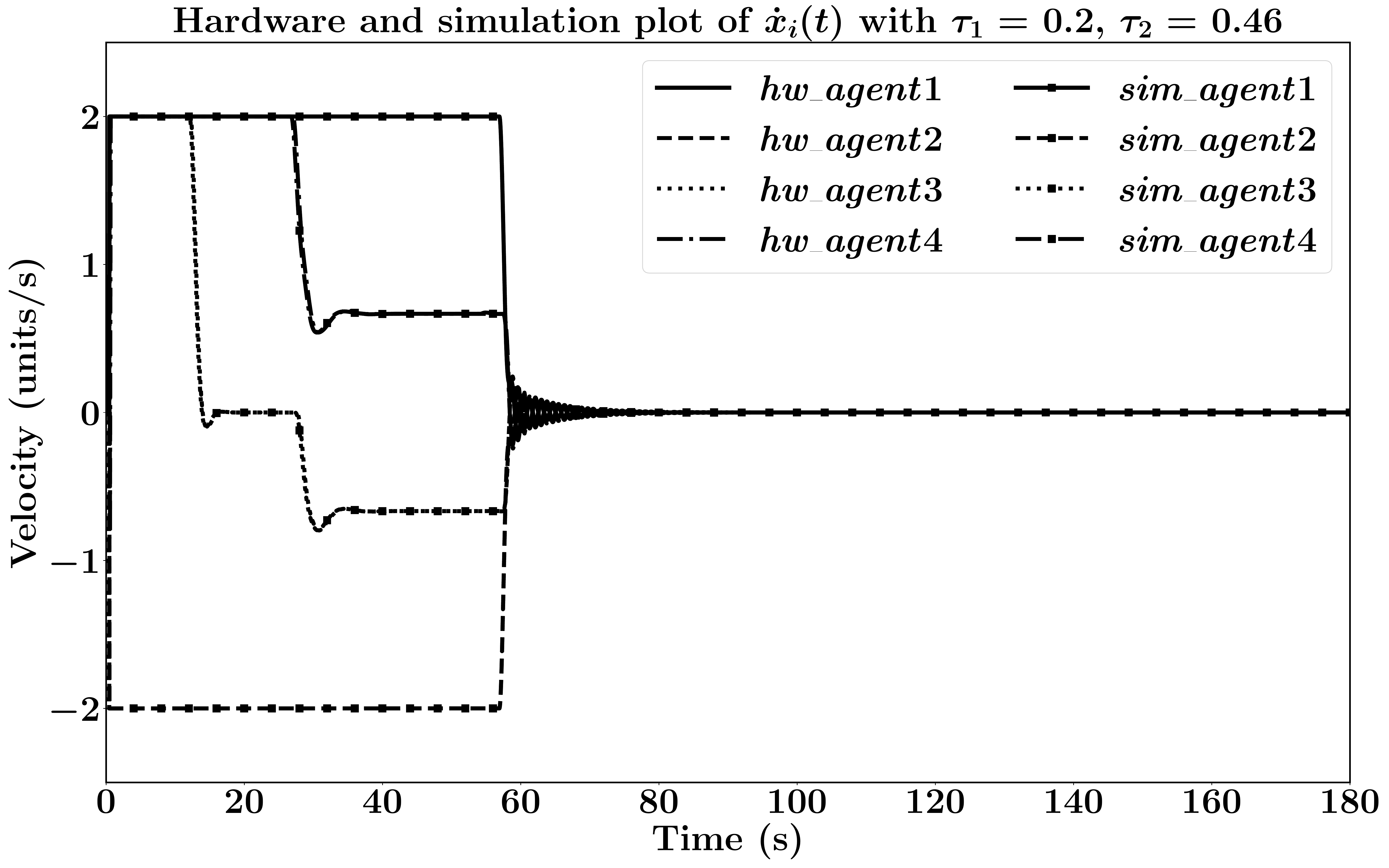}
%\caption{Plot with $u_{i4}$ in \cref{eqn1f} for non-uniform time delays.\label{res2v}}
\end{subfigure}
\caption{Plots with $u_{i4}$ in \cref{eqn1f} for asynchronous time delays.\label{res2}}
\end{figure}

\begin{figure}[!h]
\centering
\begin{subfigure}[b]{\linewidth}
\centering
\includegraphics[scale=0.14]{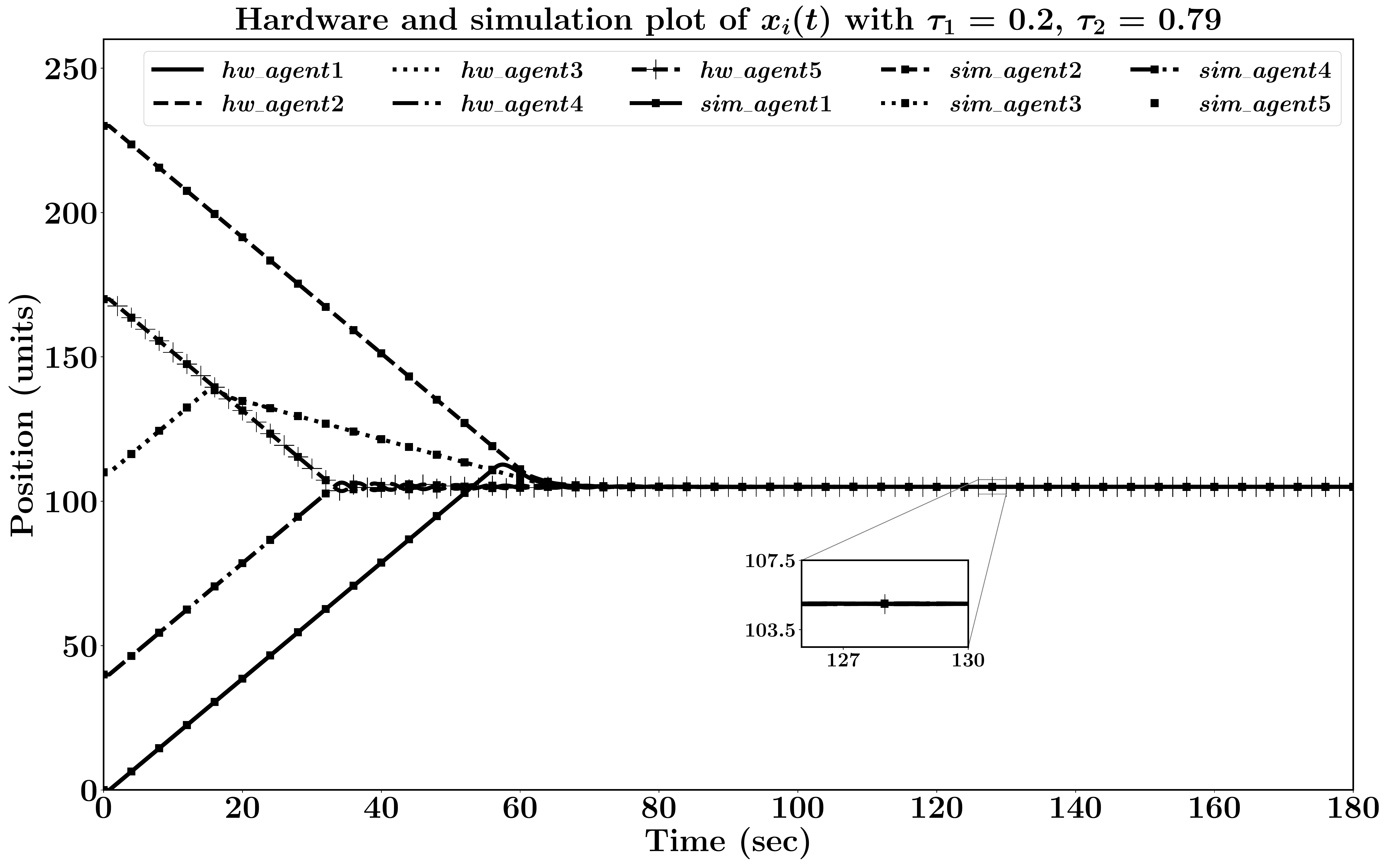}
%\caption{Plot with $u_{i1}$ in \cref{eqn2} for uneven time delays.\label{res15p}}
\end{subfigure}
\begin{subfigure}[b]{\linewidth}
\centering
\includegraphics[scale=0.14]{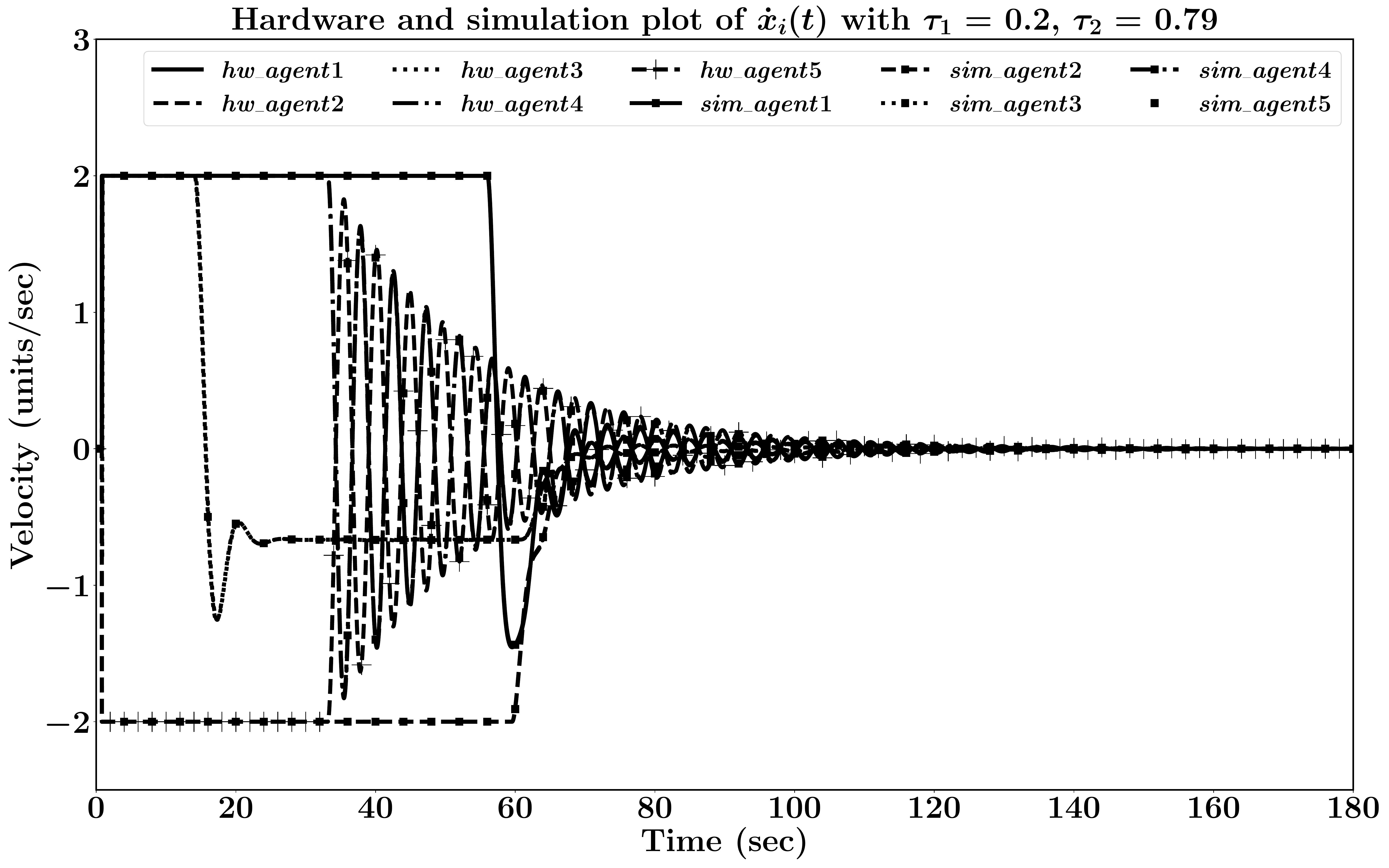}
%\caption{Plot with $u_{i1}$ in \cref{eqn2} for uneven time delays.\label{res15v}}
\end{subfigure}
\caption{Plot with $u_{i1}$ in \cref{eqn1c} for uneven time delays.\label{res15}}
\end{figure}

Similarly, some simulations and corresponding hardware validations are performed on a five-agent system with communication topology given in \cref{graph2}, the corresponding adjacency matrix $\mathcal{A}$ and $\tilde{\mathcal{A}}$ are as given below, \\~\\
$\mathcal{A}=\begin{bmatrix}0& 1& 0& 1& 0\\ 1 &0& 1& 0& 0\\ 0& 1& 0& 1& 1\\ 0& 0& 0& 0& 1\\ 0& 0& 0& 1& 0 \end{bmatrix}, 
\tilde{\mathcal{A}}=\begin{bmatrix}0& \frac{1}{2}& 0& \frac{1}{2}& 0\\ \frac{1}{2} &0& \frac{1}{2}& 0& 0\\ 0& \frac{1}{3}& 0& \frac{1}{3}& \frac{1}{3}\\ 0& 0& 0& 0& 1\\ 0& 0& 0& 1& 0\end{bmatrix} .$

The initial values for the five-agent system are considered as $x_1(0)=0$, $x_2(0)=230$, $x_3(0)=110$, $x_4(0)=40$, $x_5(0)=170$ and $\dot{x}_i(0)=0$, $i\in [1,5]$. From \cref{res15}, it can be observed that the system converges at a faster rate with $\tau_1=0.2$ and $\tau_2=0.79$. With $\tau_1=0.2$ and $\tau_2=1.03$ given by Nyquist approach, the system converges at a very slow rate in simulation and further slower in implementation due to added delay from communication links as discussed earlier, corresponding results are depicted in \cref{res1617}. Limit cycles are exhibited with $\tau_1=0.2$ and $\tau_2 \geq 1.04$ as shown in \cref{res1617}. The corresponding numerical values for Nyquist approach for $\lambda = -1$ are provided in \cref{tbl3}. 
\begin{figure}[!h]
\centering
\begin{subfigure}[b]{\linewidth}
\centering
\includegraphics[scale=0.14]{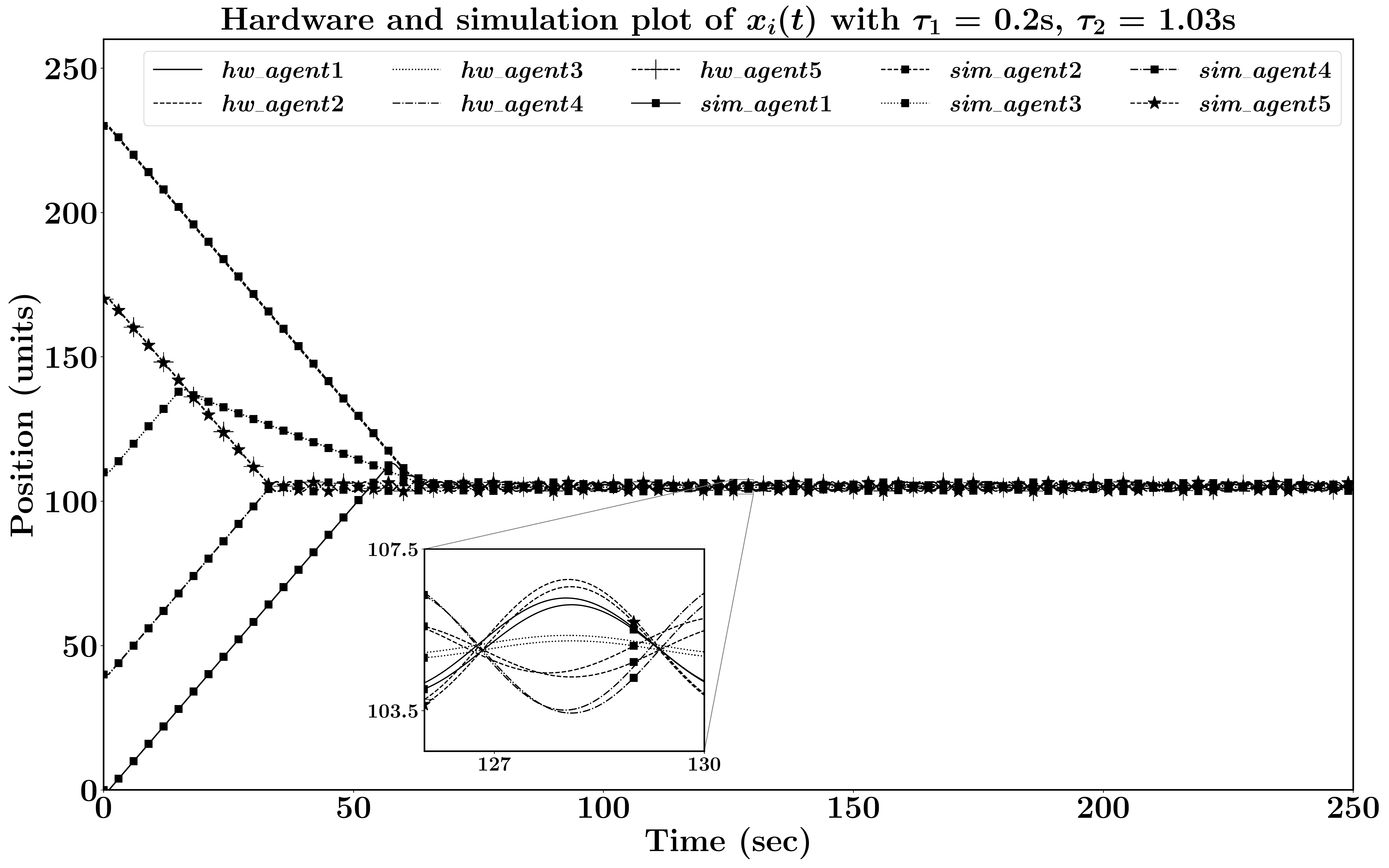}
%\caption{Plot with $u_{i1}$ in \cref{eqn2} for uneven time delays.\label{res16p}}
\end{subfigure}
\begin{subfigure}[b]{\linewidth}
\centering
\includegraphics[scale=0.14]{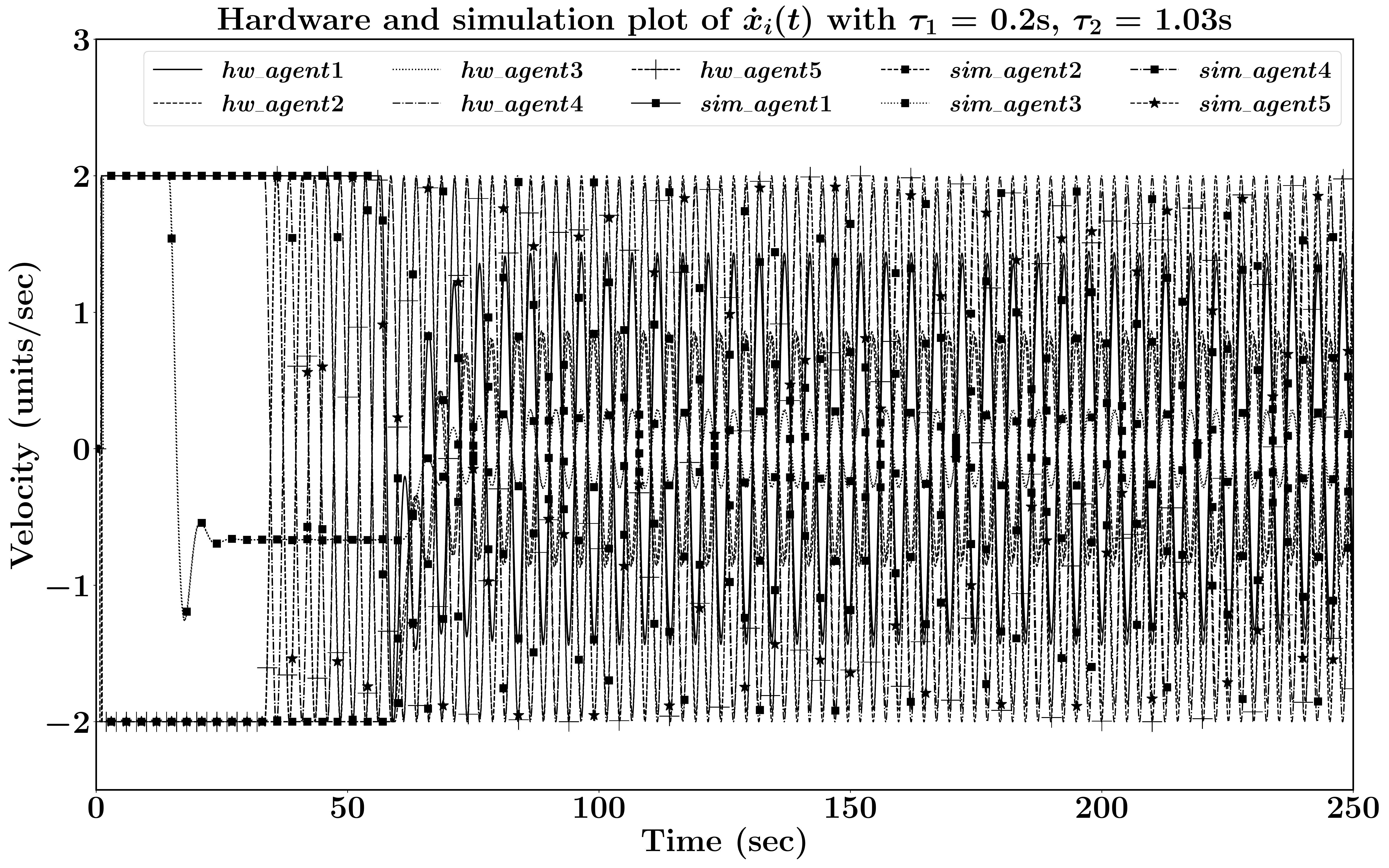}
%\caption{Plot with $u_{i1}$ in \cref{eqn2} for uneven time delays.\label{res16v}}
\end{subfigure}
\begin{subfigure}[b]{\linewidth}
\centering
\includegraphics[scale=0.14]{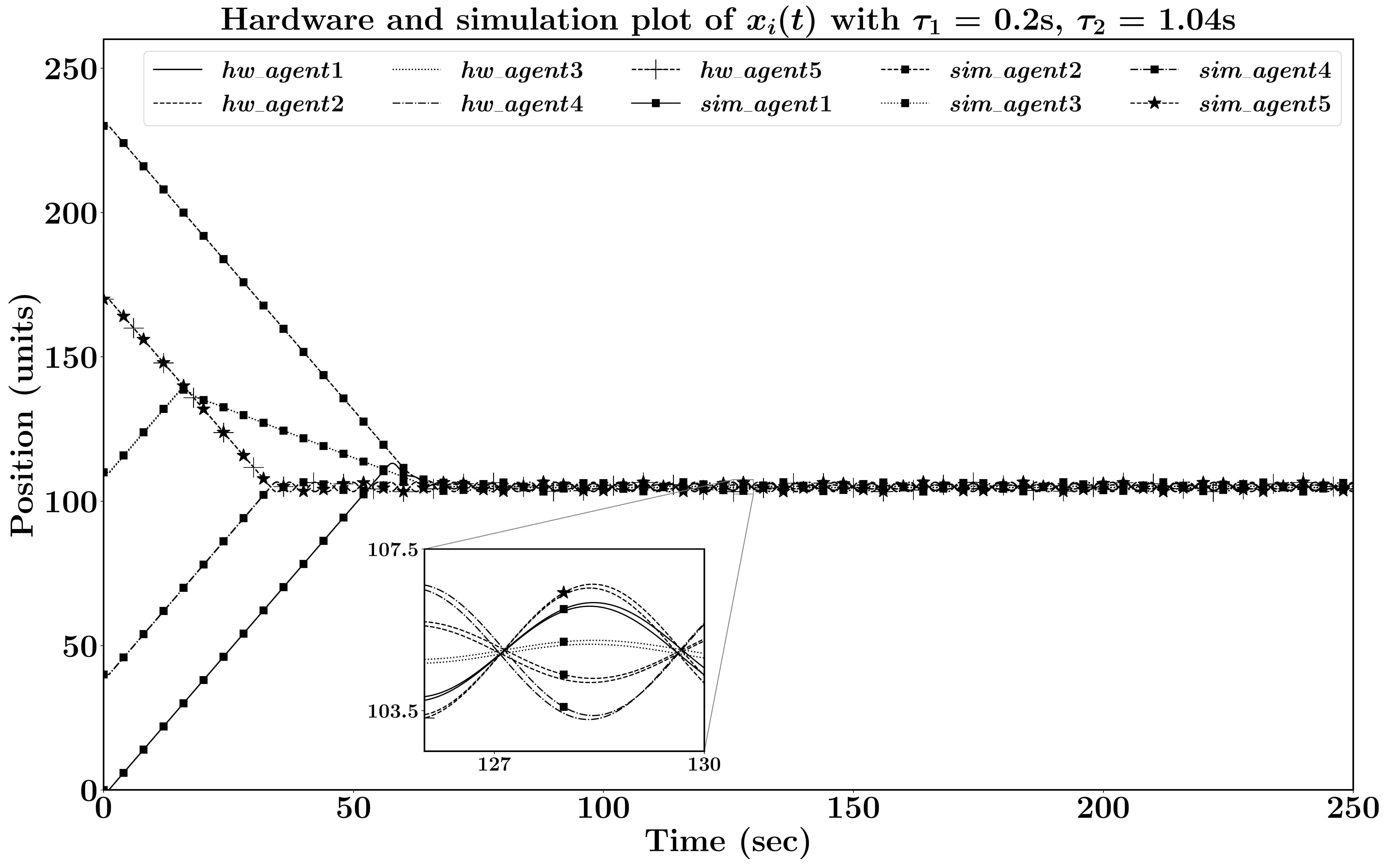}
%\caption{Plot with $u_{i1}$ in \cref{eqn2} for uneven time delays.\label{res17p}}
\end{subfigure}
\begin{subfigure}[b]{\linewidth}
\centering
\includegraphics[scale=0.14]{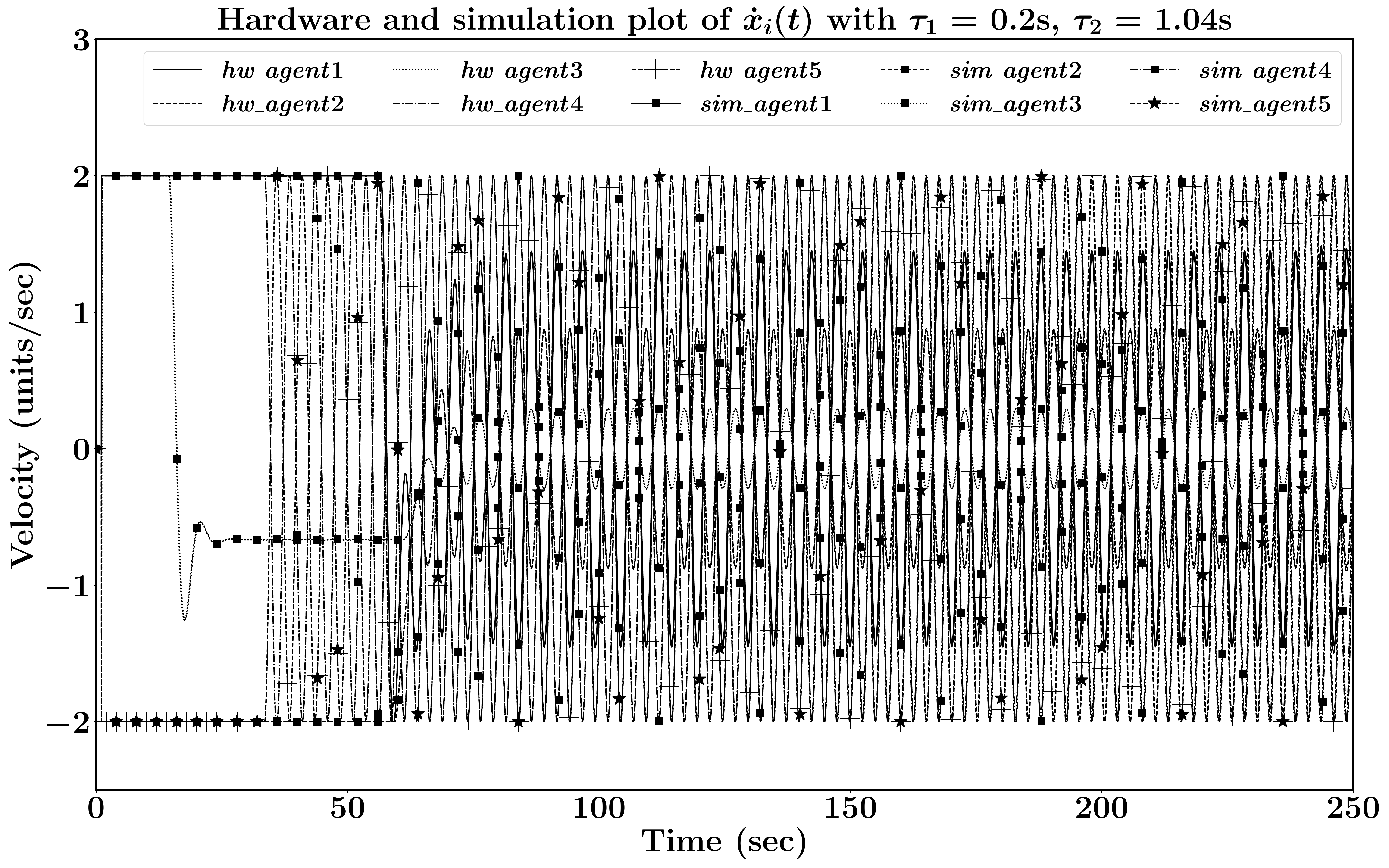}
%\caption{Plot with $u_{i1}$ in \cref{eqn2} for uneven time delays.\label{res17v}}
\end{subfigure}
\caption{Plot with $u_{i1}$ in \cref{eqn1c} for uneven time delays.\label{res1617}}
\end{figure}

Using results in \cref{th_lyp,th_freq1,th_freq2}, stable regions with respect to time-delays for both the topologies in \cref{graph1,graph2} are calculated. Few of them are validated with the help of simulations and corresponding implementations as given above. \Cref{t1_vs_t2_a,t1_vs_t2_b} depict the calculated stable regions for topologies in \cref{graph1,graph2} respectively.

\begin{table}[!h]
\processtable{Values of $\overline{\omega}$ \& $\big|G_i\left(\overline{\omega}\right)\big|$ for different values of $\tau_1$ and $\tau_2$ with $u_{i1}$ in \cref{eqn1c} and communication topology in \cref{graph2}.\label{tbl3}}
{\begin{tabular*}{20pc}{@{\extracolsep{\fill}}llll@{}}\toprule
$\tau_1$ (sec) &$\tau_2$ (sec) & $\overline{\omega}$ & $\big|G_i\left(\overline{\omega}\right)\big|$  \\ 
\midrule
  0.2 & 0.79 & 1.375 & 0.841 \\ 
  0.2 & 1.03 & 1.249 & 0.997 \\ 
  0.2 & 1.04 & 1.244 & 1.002 \\ 
\botrule
\end{tabular*}}{}
\end{table}

\subsection{Remarks}\label{remarks}
\begin{enumerate}
\item Lyapunov-Krasovskii is applicable to all the four control laws provided the communication topology satisfy the conditions mentioned in Theorem-1, but Nyquist approach is not applicable. Lyapunov-Krasovskii approach is more conservative with respect to time-delay tolerance, whereas Nyquist approach gives the full range of time-delay.
\item Describing function analysis allows us to use Nyquist approach on approximated nonlinear multi-agent system, which is better at providing time-delay tolerance ranges compared to Lyapunov approach.
\item Compared to work in \cite{Liu2013}, we have considered saturation and time-delays in the system. An approximate analysis with the help of describing function is performed. Some conditions in \cref{th_lyp,th_freq1,th_freq2} are derived for reaching consensus and estimation of non-existence of limit cycles.
\item Compared to the research presented by authors in \cite{Li2011,Meng2013,Wei2014,Chu2015,Cui2016}, we have considered time-delays along with saturation in second order multi-agent systems. Compared to the work presented by You \emph{et al.} \cite{You2016}, asynchronous time-delays are considered rather than single delay.
\end{enumerate}

\section{Conclusion}\label{cncl}
The consensus problem for second order saturated multi-agent system with asynchronous communication and input time-delays is presented in the paper. An approximate system with separate linear and nonlinear elements is derived using describing function analysis to study the limit cycle behaviour. The instability of limit cycles or consensus reachability is estimated using describing functions, stability of linear elements with the help of Lyapunov-Krasovskii function and Nyquist stability criterion. Stable ranges of input and communication time-delays are calculated for different control laws using both the approaches and comparative results are presented. Justification to the theoretical results is done with the help of simulations and corresponding implementations on hardware. With current control laws, the system is not immune to external disturbances in the state information. Noise in the state information and its mitigation strategies will be considered in the future research.
%\bibliographystyle{iet}
%\bibliography{ref_nourl}

\begin{thebibliography}{99}
\balance
\bibitem{Jadbabaie2003}
Jadbabaie, A., Lin, J., Morse, A.S.: `Coordination of groups of mobile
  autonomous agents using nearest neighbor rules', \emph{IEEE Transactions on
  Automatic Control},  2003, \textbf{48}, (6), pp.~988--1001

\bibitem{Saber2004}
Olfati.Saber, R., Murray, R.M.: `Consensus problems in networks of agents with
  switching topology and time-delays', \emph{IEEE Transactions on Automatic
  Control},  2004, \textbf{49}, (9), pp.~1520--1533

\bibitem{Vicsek1995}
Vicsek, T., Czir\'ok, A., Ben.Jacob, E., Cohen, I., Shochet, O.: `Novel type of
  phase transition in a system of self-driven particles', \emph{Phys Rev Lett},
   1995, \textbf{75}, pp.~1226--1229

\bibitem{Cao2013}
Cao, Y., Yu, W., Ren, W., Chen, G.: `An overview of recent progress in the
  study of distributed multi-agent coordination', \emph{IEEE Transactions on
  Industrial Informatics},  2013, \textbf{9}, (1), pp.~427--438

\bibitem{Wang2016}
Wang, X., Zeng, Z., Cong, Y.: `Multi-agent distributed coordination control:
  Developments and directions via graph viewpoint', \emph{Neurocomputing},
  2016, \textbf{199}, pp.~204 -- 218

\bibitem{Xiao2008}
Xiao, F., Wang, L.: `Asynchronous consensus in continuous-time multi-agent
  systems with switching topology and time-varying delays', \emph{IEEE
  Transactions on Automatic Control},  2008, \textbf{53}, (8), pp.~1804--1816

\bibitem{RenW2008}
Ren, W.: `On consensus algorithms for double-integrator dynamics', \emph{IEEE
  Transactions on Automatic Control},  2008, \textbf{53}, (6), pp.~1503--1509

\bibitem{Hu2010}
Hu, J., Lin, Y.S.: `Consensus control for multi-agent systems with
  double-integrator dynamics and time delays', \emph{IET Control Theory
  Applications},  2010, \textbf{4}, (1), pp.~109--118

\bibitem{Munz2010}
M{\"u}nz, U., Papachristodoulou, A., Allg{\"o}wer, F.: `Delay robustness in
  consensus problems', \emph{Automatica},  2010, \textbf{46}, (8), pp.~1252 --
  1265

\bibitem{Meng2011}
Meng, Z., Ren, W., Cao, Y., You, Z.: `Leaderless and leader-following consensus
  with communication and input delays under a directed network topology',
  \emph{IEEE Transactions on Systems, Man, and Cybernetics, Part B
  (Cybernetics)},  2011, \textbf{41}, (1), pp.~75--88

\bibitem{Zhang2013}
Zhang, W., Liu, J., Zeng, D., Yang, T.: `Consensus analysis of continuous-time
  second-order multi-agent systems with nonuniform time-delays and switching
  topologies', \emph{Asian Journal of Control},  2013, \textbf{15}, (5),
  pp.~1516--1523

\bibitem{Meng2016}
Meng, X., Meng, Z., Chen, T., Dimarogonas, D.V., Johansson, K.H.: `Pulse width
  modulation for multi-agent systems', \emph{Automatica},  2016, \textbf{70},
  pp.~173 -- 178

\bibitem{Yu2010}
Yu, W., Chen, G., Cao, M., Kurths, J.: `Second-order consensus for multiagent
  systems with directed topologies and nonlinear dynamics', \emph{IEEE
  Transactions on Systems, Man, and Cybernetics, Part B (Cybernetics)},  2010,
  \textbf{40}, (3), pp.~881--891

\bibitem{Liu2013}
Liu, K., Xie, G., Ren, W., Wang, L.: `Consensus for multi-agent systems with
  inherent nonlinear dynamics under directed topologies', \emph{Systems \&
  Control Letters},  2013, \textbf{62}, (2), pp.~152 -- 162

\bibitem{Li2015}
Li, J., Guan, Z.H., Chen, G.: `Multi-consensus of nonlinearly networked
  multi-agent systems', \emph{Asian Journal of Control},  2015, \textbf{17},
  (1), pp.~157--164

\bibitem{Li2011}
Li, Y., Xiang, J., Wei, W.: `Consensus problems for linear time-invariant
  multi-agent systems with saturation constraints', \emph{IET Control Theory
  Applications},  2011, \textbf{5}, (6), pp.~823--829

\bibitem{Meng2013}
Meng, Z., Zhao, Z., Lin, Z.: `On global leader-following consensus of identical
  linear dynamic systems subject to actuator saturation', \emph{Systems \&
  Control Letters},  2013, \textbf{62}, (2), pp.~132 -- 142

\bibitem{Wei2014}
Wei, A., Hu, X., Wang, Y.: `Tracking control of leader-follower multi-agent
  systems subject to actuator saturation', \emph{IEEE/CAA Journal of Automatica
  Sinica},  2014, \textbf{1}, (1), pp.~84--91

\bibitem{Chu2015}
Chu, H., Yuan, J., Zhang, W.: `Observer-based adaptive consensus tracking for
  linear multi-agent systems with input saturation', \emph{IET Control Theory
  Applications},  2015, \textbf{9}, (14), pp.~2124--2131

\bibitem{Su2015}
Su, H., Chen, M.Z.Q.: `Multi-agent containment control with input saturation on
  switching topologies', \emph{IET Control Theory Applications},  2015,
  \textbf{9}, (3), pp.~399--409

\bibitem{Cui2016}
Cui, G., Xu, S., Lewis, F.L., Zhang, B., Ma, Q.: `Distributed consensus
  tracking for non-linear multi-agent systems with input saturation: a command
  filtered backstepping approach', \emph{IET Control Theory Applications},
  2016, \textbf{10}, (5), pp.~509--516

\bibitem{You2016}
You, X., Hua, C., Peng, D., Guan, X.: `Leader following consensus for
  multi-agent systems subject to actuator saturation with switching topologies
  and time-varying delays', \emph{IET Control Theory Applications},  2016,
  \textbf{10}, (2), pp.~144--150

\bibitem{Sun2008}
Sun, Y.G., Wang, L., Xie, G.: `Average consensus in networks of dynamic agents
  with switching topologies and multiple time-varying delays', \emph{Systems \&
  Control Letters},  2008, \textbf{57}, (2), pp.~175 -- 183

\bibitem{Lin2010}
Lin, P., Jia, Y.: `Consensus of a class of second-order multi-agent systems
  with time-delay and jointly-connected topologies', \emph{IEEE Transactions on
  Automatic Control},  2010, \textbf{55}, (3), pp.~778--784

\bibitem{slotine_chap}
Slotine, J.J.E., Li, W.
\newblock 5, Describing Function Analysis.
\newblock In: `Applied nonlinear control'. (Englewood Cliffs (N.J.): Prentice
  Hall,  1991. pp.~ 157--190

\bibitem{Lin2008}
Lin, P., Jia, Y.: `Average consensus in networks of multi-agents with both
  switching topology and coupling time-delay', \emph{Physica A: Statistical
  Mechanics and its Applications},  2008, \textbf{387}, (1), pp.~303 -- 313

\bibitem{sedumi}
Sturm, J.F.: `Using {SeDuMi} 1.02, a {MATLAB} toolbox for optimization over
  symmetric cones', \emph{Optimization Methods and Software},  1999,
  \textbf{11--12}, pp.~625--653.
\newblock version 1.05 available from {\texttt{http://fewcal.kub.nl/sturm}}

\end{thebibliography}

\end{document}